\documentclass[a4paper]{article}
\usepackage{mlcv}
\usepackage{bm}

\Crefname{myclaim}{Claim}{Claims}

\newcommand{\Z}{\mathbb{Z}}
\newcommand{\N}{\mathbb{N}}
\newcommand{\R}{\mathbb{R}}
\newcommand{\bb}{\mathbbm{1}}
\newcommand{\set}[1]{{\{#1\}}}
\DeclareMathOperator{\cpp}{CP}
\DeclareMathOperator{\cyc}{cyc}

\newcommand*{\myproofname}{Proof of claim}

\usetikzlibrary{decorations.pathreplacing}
\usetikzlibrary{calc}
\tikzstyle{vertex}=[circle, draw, fill=white, inner sep=0.05cm, minimum width=1ex, align=center]
\tikzstyle{cut}=[dashed]
\tikzstyle{area2} = [draw=black!25!white, line width=0.28cm, line cap=round, line join=round]
\tikzstyle{area3} = [draw=black!10!white, line width=0.43cm, line cap=round, line join=round]
\tikzstyle{arealight2} = [draw=black!10!white, fill=black!10!white, line width=0.43cm]

\begin{document}

\title{\textbf{Chorded cycle facets of the clique partitioning polytope}}
\author{Jannik Irmai$^a$, Lucas Fabian Naumann$^a$, Bjoern Andres$^{a,b,1}$}
\date{\begin{minipage}{\columnwidth}\centering
	$^a$Faculty of Computer Science, TU Dresden, Germany\\
	$^b$Center for Scalable Data Analytics and AI (ScaDS.AI) Dresden/Leipzig, Germany\\
	\end{minipage}}

\footnotetext[1]{Correspondence: \texttt{bjoern.andres@tu-dresden.de}}
\setcounter{footnote}{1}

\maketitle

\begin{abstract}
    The $q$-chorded $k$-cycle inequalities are a class of valid inequalities for the clique partitioning polytope. It is known that for $q \in \{2, \tfrac{k-1}{2}\}$, these inequalities induce facets of the clique partitioning polytope if and only if $k$ is odd. Here, we characterize such facets for arbitrary $k$ and $q$. More specifically, we prove that the $q$-chorded $k$-cycle inequalities induce facets of the clique partitioning polytope if and only if two conditions are satisfied: $k = 1$ mod $q$, and if $k=3q+1$ then $q=3$ or $q$ is even. This establishes the existence of many facets induced by $q$-chorded $k$-cycle inequalities beyond those previously known.

\end{abstract}

\section{Introduction}
Given a complete graph with edge values that can be both positive and negative real numbers, the \emph{clique partitioning problem} consists in finding a partition of the graph into disjoint cliques that maximizes the value of the edges within the cliques.
This problem has a wide range of applications, including the aggregation of binary relations \cite{grotschel1989cutting}, community detection in social networks \cite{brandes2007modularity}, and group technology \cite{oosten2001clique}.
Its feasible solutions are encoded by binary vectors with one entry for each edge of the graph, where an entry is 1 if and only if the associated edge is contained in a clique.
The convex hull of these vectors is called the \emph{clique partitioning polytope} \cite{grotschel1990facets}.
While a complete outer description of this polytope in terms of its facets is not known, many classes of valid and facet-inducing inequalities have been discovered and are described in the literature; see \Cref{section:related-work}.
From a practical perspective, such valid inequalities are essential for strengthening linear programming bounds and thus accelerating branch-and-cut algorithms.
One such class of valid inequalities is that of the $2$-chorded $k$-cycle inequalities introduced in \cite{grotschel1990facets} and shown to induce a facet if and only if the cycle has odd length $k$.
This class is generalized in \cite{muller2002transitive} to the $q$-chorded $k$-cycle inequalities, where $q$ is any integer between $2$ and $\tfrac{k}{2}$.
These inequalities are valid for the clique partitioning problem \cite{muller2002transitive}, but no claims have been made about facets induced by these inequalities.
Recently, it was shown that for the special case of $q=\tfrac{k-1}{2}$, the $q$-chorded $k$-cycle inequalities induce facets of the clique partitioning polytope \cite{andres2023polyhedral}.

In this article, we establish for arbitrary $k$ and $q$ the exact condition under which the $q$-chorded $k$-cycle inequalities induce facets of the clique partitioning polytope (\Cref{thm:main-theorem}).
For the special cases of $q=2$ and $q=\tfrac{k-1}{2}$, this condition specializes to the properties previously known.
In its general form, it implies the existence of many facets induced by $q$-chorded $k$-cycle inequalities for $2 < q < \tfrac{k-1}{2}$ previously unknown.
\section{Related Work}
\label{section:related-work}

The clique partitioning problem is closely related to \emph{coalition structure generation in weighted graph games} \cite{voice2012coalition,bachrach2013optimal}, the \emph{multicut problem} \cite{deza1992clique,chopra1993partition}, and \emph{correlation clustering} \cite{bansal2004correlation,demaine2006correlation}.
For complete graphs, these problems are equivalent in the sense that there are bijections between their sets of optimal solutions.
However, they differ with regard to the hardness of approximation.
The clique partitioning problem is \textsc{np}-hard to approximate within a factor $\mathcal{O}(n^{1-\epsilon})$, where $n$ is the number of nodes and $\epsilon$ is any positive constant \cite{bachrach2013optimal,zuckerman2006linear}.
Exact algorithms based on cutting plane methods are discussed, e.g. in \cite{grotschel1989cutting,oosten2001clique,sorensen2020separation}.

The clique partitioning polytope is introduced and studied in detail for the first time in \cite{grotschel1990facets}.
Today, many classes of valid inequalities are established, along with conditions under which they induce facets.
Examples include $2$-chorded cycle, path, and even wheel inequalities \cite{grotschel1990facets}, (generalized) $2$-partition inequalities \cite{grotschel1990facets,grotschel1990composition}, (bicycle) wheel inequalities \cite{chopra1993partition}, clique-web inequalities \cite{deza1992clique}, hypermetric inequalities \cite{deza1997geometry}, and further generalizations of the preceding inequalities \cite{oosten2001clique}.
Techniques for deriving additional facets from known facets are studied in \cite{grotschel1990composition,deza1992clique,chopra1995facets,oosten2001clique}.
For instance, it has been conjectured in \cite{grotschel1990facets} and proven independently in \cite{deza1992clique,chopra1995facets,bandelt1999lifting} that \emph{zero-lifting} holds for the clique partitioning polytope; see \Cref{thm:zero-lifting} below.
However, even for the complete graph with only $n=5$ nodes, the clique partitioning polytope has many facets that have not yet been characterized \cite{deza1991complete}.

Of particular interest here are the \emph{$2$-chorded cycle} inequalities defined in \cite{grotschel1990facets} with respect to a cycle of length $k$ and its $2$-chords.
They induce facets if and only if the cycle is odd \cite{grotschel1990facets}.
Of similar interest here are the \emph{half-chorded odd cycle} inequalities defined in \cite{andres2023polyhedral} with respect to a cycle of odd length $k$ and its $\tfrac{k-1}{2}$-chords.
They all induce facets \cite{andres2023polyhedral}.
Both belong to the more general class of \emph{$q$-chorded cycle} inequalities discovered and shown to be valid for the clique partitioning polytope in \cite{muller2002transitive} using ideas introduced there that unify the study of various combinatorial optimization problems over transitive structures, including clique partitioning.
In our characterization of the facets induced by these inequalities, the length $k$ of the cycle plays a crucial role.
Thus, we refer to these inequalities as $q$-chorded $k$-cycle inequalities throughout this article.

A closely related class of inequalities that is also defined with respect to chorded cycles is discussed in \cite{poljak1992max} for the cut polytope.
These inequalities have different coefficients than those discussed here as the cut polytope is defined with respect to partitions into exactly two sets in contrast to an arbitrary number of sets for clique partitioning.
\section{Preliminaries}
In this section, we recall the definition of the clique partitioning problem and polytope and state some basic properties.

For any $n \in \N$, let $K_n = (V_n, E_n)$ be the complete undirected graph with $n$ nodes, i.e.~$V_n = \{0,\dots,n-1\}$ and $E_n = \binom{V_n}{2}$.
An edge subset $A \subseteq E_n$ is called a \emph{clique partition} of $K_n$ if and only if there exists a partition $\Pi$ of the nodes $V_n$ such that $A$ contains precisely those edges that connect nodes that are in the same set of $\Pi$, i.e.~$A = \left\{ \{i,j\} \in E \mid \exists U \in \Pi: i,j \in U \right\}$. 
Clearly, there exists a one-to-one relation between the partitions of $V_n$ and the clique partitions of $K_n$.
For a partition $\Pi$ of $V_n$, let $x^\Pi = \bb_{A}$ be the characteristic vector of the clique partition $A$ that is associated with $\Pi$.
I.e., $x^\Pi_{\{i,j\}} = 1$ if and only if there exists $U \in \Pi$ with $i,j \in U$ for all $\{i,j\} \in E_n$.
We call $x^\Pi$ the feasible vector \emph{induced by} the partition $\Pi$.
The following lemma characterizes the characteristic vectors of clique partitions in terms of triangle inequalities.

\begin{lemma}[\cite{grotschel1990facets}]
    A binary vector $x \in \{0,1\}^{E_n}$ is the characteristic vector of a clique partition if and only if 
    \begin{align*}
        x_{\{i,j\}} + x_{\{j, k\}} - x_{\{i,k\}} \leq 1 \quad \text{for all pairwise distinct } i,j,k \in V_n \enspace.
    \end{align*}
\end{lemma}

Let $c:E_n \to \R$ assign a value to each edge of $K_n$.
The \emph{clique partitioning problem} consists in finding a clique partition of $K_n$ of maximal value.
It has the form of the integer linear program
\begin{align}
    \max \quad & \sum_{\{i,j\} \in E_n} c_{\{i,j\}} x_{\{i,j\}} \notag \\
    \text{s.t.} \quad & x \in \Z^{E_n} \notag \\
    & 0 \leq x_{\{i,j\}} \leq 1 && \text{for all } \{i,j\} \in E_n \label{eq:box} \\
    & x_{\{i,j\}} + x_{\{j, k\}} - x_{\{i,k\}} \leq 1 && \text{for all pairwise distinct } i,j,k \in V_n \enspace. \label{eq:triangle}
\end{align}

The \emph{clique partitioning polytope} is defined as the convex hull of all feasible solutions of the clique partitioning problem and is denoted by
\[
    \cpp_n = \text{conv}\left\{ x \in \{0,1\}^{E^n} \;\middle|\; x \text{ satisfies } \eqref{eq:triangle}\right\} \enspace. 
\]
The vertices of this polytope are precisely the characteristic vectors of clique partitions.
As observed in \cite{grotschel1990facets}, the clique partitioning polytope $\cpp_n$ is full dimensional, i.e.~$\dim \cpp_n = |E_n| = \binom{n}{2}$, as it contains the zero vector and all unit vectors.

The following theorem states that an inequality that induces a facet of one clique partitioning polytope also induces a facet of every larger clique partitioning polytope.

\begin{theorem}[Zero Lifting \cite{deza1992clique,chopra1995facets,bandelt1999lifting}]\label{thm:zero-lifting}
    Let $n \in \N$, $a \in \Z^{E_n}$, and $\alpha \in \Z$ such that $a^\top x \leq \alpha$ induces a facet of $\cpp_n$.
    Then, for every $m \in \N$, $m \geq n$, the \emph{lifted inequality} $\bar{a}^\top \bar{x} \leq \alpha$ induces a facet of $\cpp_m$ where $\bar{a}_{\{i,j\}} = a_{\{i,j\}}$ if $\{i,j\} \in E_n$, and $\bar{a}_{\{i,j\}} = 0$ otherwise, for all $\{i,j\} \in E_m$.
\end{theorem}
\section{Chorded cycle inequalities}
\begin{figure}[t]
    \centering
    \begin{tikzpicture}[baseline=0]
    \def \r{1.2cm}
    \def \q{3}
    \def \p{4}
    \pgfmathtruncatemacro{\n}{\p*\q+1}
    \pgfmathtruncatemacro{\nn}{\n-1}
    \def \a{360/\n}

    \foreach \i in {0,...,\nn}
    {
        \node[vertex] (\i) at (\i*\a:\r) {};
    }
    \foreach \i in {0,...,\nn}{
        \pgfmathtruncatemacro{\j}{mod(\i+1,\n)}%
        \pgfmathtruncatemacro{\k}{mod(\i+\q,\n)}%
        \draw (\i) -- (\j);
        \draw[dashed] (\i) -- (\k);
    }
\end{tikzpicture}
    \hfill
    \begin{tikzpicture}[baseline=0]
    \def \r{1.2cm}
    \def \rs{0.1*\r}
    \def \q{3}
    \def \p{4}
    \pgfmathtruncatemacro{\n}{\p*\q+1}
    \pgfmathtruncatemacro{\nn}{\n-1}
    \def \a{360/\n}

    \draw[arealight2] plot [smooth cycle, tension=0.2] coordinates {
    	(3*\a:\r) 
    	(4*\a:\r) 
    	(7*\a:\rs)
    	(10*\a:\r) 
    	(12*\a:\r) 
    	(\a:\rs)
   	};

    \foreach \style in {area3, area2}
    {
	    \draw[\style] (0:\r) arc (0:2*\a:\r);
	    \draw[\style] (3*\a:\r) arc (3*\a:4*\a:\r);
	    \draw[\style] (5*\a:\r) arc (5*\a:7*\a:\r);
	    \draw[\style] (8*\a:\r) arc (8*\a:9*\a:\r);
	    \draw[\style] (10*\a:\r) arc (10*\a:12*\a:\r);
    }

    \foreach \i in {0,...,\nn}
    {
        \node[vertex] (\i) at (\i*\a:\r) {};
    }
\end{tikzpicture}
    \hfill
    \begin{tikzpicture}[baseline=0]
    \def \r{1.2cm}
    \def \rs{0.1*\r}
    \def \q{3}
    \def \p{4}
    \pgfmathtruncatemacro{\n}{\p*\q+1}
    \pgfmathtruncatemacro{\nn}{\n-1}
    \def \a{360/\n}

    \draw[arealight2] plot [smooth cycle, tension=0.2] coordinates {
		(0*\a:\r) 
    	(3*\a:\r) 
		(5*\a:\rs)
    	(7*\a:\r) 
    	(9*\a:\r) 
		(11*\a:\rs)
	};

    \foreach \style in {area3, area2}
    {
	    \draw[\style] (0:\r) arc (0:3*\a:\r);
	    \draw[\style] (4*\a:\r) arc (4*\a:6*\a:\r);
	    \draw[\style] (7*\a:\r) arc (7*\a:9*\a:\r);
	    \draw[\style] (10*\a:\r) arc (10*\a:12*\a:\r);
	}

    \foreach \i in {0,...,\nn}
    {
        \node[vertex] (\i) at (\i*\a:\r) {};
    }
\end{tikzpicture}
    \hfill
    \begin{tikzpicture}[baseline=0]
    \def \r{1.2cm}
    \def \rs{0.1*\r}
    \def \q{3}
    \def \p{4}
    \pgfmathtruncatemacro{\n}{\p*\q+1}
    \pgfmathtruncatemacro{\nn}{\n-1}
    \def \a{360/\n}

    \draw[arealight2] plot [smooth cycle, tension=0.2] coordinates {
    	(3*\a:\r) 
    	(4*\a:\r) 
    	(6*\a:\rs) 
    	(8*\a:\r) 
    	(9*\a:\r) 
    	(10*\a:\rs) 
    	(11.9*\a:\r) 
    	(12.1*\a:\r) 
    	(0*\a:\rs) 
   	};

    \foreach \style in {area3, area2}
    {
	    \draw[\style] (0:\r) arc (0:2*\a:\r);
	    \draw[\style] (3*\a:\r) arc (3*\a:4*\a:\r);
	    \draw[\style] (5*\a:\r) arc (5*\a:7*\a:\r);
	    \draw[\style] (8*\a:\r) arc (8*\a:9*\a:\r);
	    \draw[\style] (10*\a:\r) arc (10*\a:11*\a:\r);
	    \draw[\style] (12*\a:\r) arc (12*\a:12.001*\a:\r);
	}

    \foreach \i in {0,...,\nn}
    {
        \node[vertex] (\i) at (\i*\a:\r) {};
    }
\end{tikzpicture}
    \caption{Depicted on the left is the graph associated with the $q$-chorded $k$-cycle inequality for $k=13$ and $q=3$ where the solid lines represent edges $\{i,i+1\}$ and the dashed lines represent $q$-chords $\{i,i+q\}$ for $i \in \Z_k$.
    Depicted on the right are three partitions of the $k$ nodes whose components are indicated by light shaded areas.
    The respective induced cycle partitions are depicted by a darker shade.
    For example, the first partition consists of $4$ components and the induced cycle partition has $5$ components.
    The first two partitions induce feasible solutions that satisfy the $q$-chorded $k$-cycle inequality at equality while the third partition does not.
    The first partition satisfies \ref{item:eq-a} of \Cref{lem:part}, the second partition satisfies \ref{item:eq-b}, and the third partition satisfies neither.
    In all subsequent figures, we will keep the introduced convention of depicting partitions by light shaded areas and induced cycle partitions by darker shaded areas.
    }
    \label{fig:cycle-partition}
\end{figure}

In this section, we recall from \cite{muller2002transitive} the class of $q$-chorded $k$-cycle inequalities and establish a necessary and sufficient condition for such an inequality to induce a facet of the clique partitioning polytope.
Thanks to the zero lifting theorem (\Cref{thm:zero-lifting}), we can restrict the analysis of an inequality whose support graph consists of $k$ nodes to the polytope $\cpp_k$.

For our discussion involving cycles, it will be convenient to identify the nodes of the graph with the integers modulo $k$, i.e.~$V_k = \Z_k$.
This will allow us, e.g.~to write the edge set $E = \left\{\{0,1\},\{1,2\},\dots,\{k-2,k-1\},\{0, k-1\}\right\}$ of a cycle more conveniently as $E=\left\{\{i,i+1\} \mid i \in \Z_k\right\}$.
Given the cycle defined by this edge set $E$, an edge $\{i,i+p\}$ with $i,p \in \Z_k$, $p \neq 0$ is called a \emph{$p$-chord} of that cycle.

\begin{definition}[\cite{muller2002transitive}]
    For any $q, k \in \N$ with $2 \leq q \leq k / 2$, the $q$-chorded $k$-cycle inequality is defined as
    \begin{align}\label{eq:q-chored-k-cycle}
        \sum_{i \in \Z_k} \left( x_{\{i,i+1\}} - x_{\{i, i+q\}} \right) \leq k - \left\lceil \tfrac{k}{q} \right\rceil \enspace .
    \end{align}
\end{definition}

An example is depicted on the left of \Cref{fig:cycle-partition}.
The validity of the $q$-chorded $k$-cycle inequalities for the clique partitioning polytope is proven in \cite{muller2002transitive} in a more general context.
For completeness, we include a simple proof of the following lemma in \Cref{app:proof-valid}.

\begin{lemma}\label{lem:valid}
    For any $q, k \in \N$ with $2 \leq q \leq k / 2$, the $q$-chorded $k$-cycle inequality \eqref{eq:q-chored-k-cycle} is a Chv\'atal-Gomory cut with respect to the system of box \eqref{eq:box} and triangle \eqref{eq:triangle} inequalities.
    In particular, it is valid for the clique partitioning polytope $\cpp_k$.
\end{lemma}

The following definition and lemma serve the purpose of characterizing those vertices of the clique partitioning polytope that satisfy the $q$-chorded $k$-cycle inequalities at equality.
An example is given in \Cref{fig:cycle-partition}.
Subsequently, we will be able to state and prove the main theorem of this article.

\begin{definition}
    Let $k \in \N$ and let $\Pi$ be a partition of $\Z_k$.
    Let $\Pi^{\cyc}$ denote the node sets of the connected components of the graph $(\Z_k, E)$ with $E = \{\{i,i+1\} \mid i \in \Z_k \text{ and } \exists U \in \Pi: i,i+1 \in U\}$.
    We call $\Pi^{\cyc}$ the \emph{cycle partition induced by $\Pi$}.
\end{definition}

We remark that the partition $\Pi$ can be obtained from its induced cycle partition $\Pi^{\cyc}$ by joining some non-adjacent components.

\begin{lemma}\label{lem:part}
    Let $q, k \in \N$ with $2 \leq q \leq k / 2$ and let $\Pi$ be a partition of $\Z_k$.
    Then, the feasible solution $x^\Pi$ associated with $\Pi$ satisfies \eqref{eq:q-chored-k-cycle}
    at equality if and only if one of the two following sets of conditions is satisfied
    \begin{enumerate}[label = {(\alph*)}]
        \item \label{item:eq-a}
        \begin{enumerate}[label = {(a\arabic*)}]
            \item \label{item:eq-a1}
             $|\Pi^{\cyc}| = \lceil k / q \rceil$, and
            \item \label{item:eq-a2}
             $|U| \leq q$ for all $U \in \Pi^{\cyc}$, and
            \item \label{item:eq-a3}
             $i-j>q$ for all $U \in \Pi$, for all $U_1,U_2 \in \Pi^{\cyc}$ with $U_1 \neq U_2$ and $U_1, U_2 \subseteq U$, and for all $i \in U_1, j \in U_2$.
        \end{enumerate}
        \item \label{item:eq-b}
        \begin{enumerate}[label = {(b\arabic*)}]
            \item \label{item:eq-b1}
            $|\Pi^{\cyc}| = \lceil k / q \rceil - 1$, and 
            \item \label{item:eq-b2}
            there exists $U \in \Pi^{\cyc}$ with $|U| = q+1$  and $|U'| = q$ for all $U' \in \Pi^{\cyc} \setminus \{U\}$.
        \end{enumerate}
    \end{enumerate}
\end{lemma}
The proof of \Cref{lem:part} is deferred to \Cref{app:proof-part}.

\begin{theorem}\label{thm:main-theorem}
	For any $q, k \in \N$ with $2 \leq q \leq k / 2$, the $q$-chorded $k$-cycle inequality induces a facet of the clique partitioning polytope if and only if the following two conditions both hold
    \begin{enumerate}[label = {(\roman*)}]
        \item \label{item:condition-1} $k = 1$ mod $q$ 
        \item \label{item:condition-2} if $k = 3q + 1$ then $q = 3$ or $q$ is even.
    \end{enumerate}
\end{theorem}

\begin{proof}
    \textbf{Necessity.}
    We begin by showing the necessity of both conditions.
    More specifically, we show that if any of the conditions is violated then the induced face is not a facet of the polytope.
    To do so, we establish an equality independent of \eqref{eq:q-chored-k-cycle} that is satisfied by all feasible solutions in the face.
    This implies that the face has dimension at most $\binom{k}{2}-2$ and is thus not a facet, as $\cpp_k$ has dimension $\binom{k}{2}$.

    Suppose \ref{item:condition-1} is violated.
    Let $\Pi$ be a partition of $\Z_k$ such that $x=x^\Pi$ satisfies \eqref{eq:q-chored-k-cycle} at equality.
    By \Cref{lem:part}, $\Pi$ satisfies either \ref{item:eq-a} or \ref{item:eq-b}.
    As \ref{item:eq-b} cannot hold by the assumption that $k \neq 1$ mod $q$, \ref{item:eq-a} holds.
    By $|\Pi^{\cyc}| = \lceil k / q \rceil$ and the definition of $\Pi^{\cyc}$, we get
    \[
        \sum_{i \in \Z_k} x_{\{i,i+1\}} = \lceil k / q \rceil \enspace .
    \]
    As this equality holds for all $x$ in the face induced by \eqref{eq:q-chored-k-cycle}, this face cannot be a facet.

    Next, suppose \ref{item:condition-1} holds and \ref{item:condition-2} is violated, i.e.~$q > 3$ odd and $k = 3q+1$.
    We show that in this case the additional equality
    \begin{align}\label{eq:second-condition-equality}
        \sum_{i \in \Z_k} (-1)^i x_{\{i,i+q+2\}} = 0
    \end{align}
    holds for all $x$ in the face induced by \eqref{eq:q-chored-k-cycle}.
    Note that, because $q > 3$, all edges in \eqref{eq:second-condition-equality} are distinct. (For $q = 3$, we have $q + 2 = - q - 2$ (mod $k$), and \eqref{eq:second-condition-equality} reduces to the trivial equality $0 = 0$.)

    Let $\Pi$ be a partition of $\Z_k$ such that $x=x^\Pi$ satisfies \eqref{eq:q-chored-k-cycle} at equality.
    If $x_{\{i,i+q+2\}} = 0$ for all $i \in \Z_k$, then \eqref{eq:second-condition-equality} clearly holds.
    Now, suppose there exists $i \in \Z_k$ with $x_{\{i,i+q+2\}} = 1$.
    Then, \ref{item:eq-b} of \Cref{lem:part} cannot hold by the following argument:
    If \ref{item:eq-b} holds, then we have $|\Pi^{\cyc}| =  \lceil k / q \rceil - 1 = 3$ and $|U| \leq q+1$ for all $U \in \Pi^{\cyc}$.
    By definition of $\Pi^{\cyc}$ and by $|\Pi^{\cyc}| = 3$: $\Pi = \Pi^{\cyc}$.
    By $|U| \leq q+1$ and by definition of $\Pi^{\cyc}$: $\{i,i+q+2\} \not\subseteq U$ for all $U \in \Pi$ and all $i \in \Z_k$.
    This implies $x_{\{i,i+q+2\}} = 0$ for all $i \in \Z_k$, in contradiction to the assumption.
    Therefore, \ref{item:eq-a} of \Cref{lem:part} holds.
    Since $\lceil k / q \rceil = 4$, Conditions \ref{item:eq-a1} and \ref{item:eq-a2} state that $|\Pi^{\cyc}| = 4$ with $|U| \leq q$ for all $U \in \Pi^{\cyc}$.
    And since $\{i,i+q+2\} \not\subseteq U$ for all $U \in \Pi^{\cyc}$, at least two sets in $\Pi^{\cyc}$ must be joined in $\Pi$.
    Due to Condition \ref{item:eq-a3}, this can only be the case if $\Pi^{\cyc} = \{P_1,Q_1,P_2,Q_2\}$, $\Pi = \{P_1 \cup P_2, Q_1, Q_2\}$ with $|P_1| = p$, $|P_2| = q + 1 - p$, and $|Q_1| + |Q_2| = q$ for some $p \in \{1,\dots,\tfrac{q+1}{2}\}$.
    Now, let $i \in \Z_k$ such that $P_1 = \{i,i+1,\dots,i+p-1\}$.
    By construction, $P_2 = \{i+p+q,\dots,i-q-1\}$.
    An illustration can be found in \Cref{fig:necessity-2}.

    If $p = 1$, i.e.~$P_1 = \{i\}$, then $x_{\{i, i+q+2\}} = x_{\{i-q-2, i\}} = 1$ and $x_{\{j, j+q+2\}} = 0$ for all $j \in \Z_k \setminus \{i, i-q-2\}$.
    As $q$ is odd, $k=3q+1$ is even.
    Thus, either $i$ is odd or $i-q-2$ (mod $k$) is odd and, therefore, $x$ satisfies \eqref{eq:second-condition-equality}.

    Now, consider the case $p \geq 2$.
    Let $A = \{i-q-2, i-q-1, i+p-2, i+p-1\}$.
    Then, by construction, $x_{\{j, j+q+2\}} = 1$ for all $j \in A$ and $x_{\{j, j+q+2\}} = 0$ for all $j \in \Z_k \setminus A$.
    By a similar argument as above, exactly two values in $A$ are odd and two values in $A$ are even.
    Therefore, $x$ satisfies \eqref{eq:second-condition-equality}.

    \begin{figure}[t]
        \centering
        \begin{tikzpicture}[baseline=0]
    \def \r{1.7cm}
    \def \rs{0.2*\r}
    \def \l{1.9cm}
    \def \b{0.3cm}
    \def\n{16}
    \pgfmathtruncatemacro{\nn}{\n-1}
    \def \a{360/\n}

    \draw[arealight2] plot [smooth cycle, tension=0.2] coordinates {
		(2*\a:\r) 
        (4*\a:\r)
		(6*\a:\r) 
		(8*\a:\rs)
		(12*\a:\r)
		(0*\a:\rs)
    };

    \foreach \style in {area3, area2}
    {
	    \draw[\style] (2*\a:\r) arc (2*\a:6*\a:\r);
	    \draw[\style] (7*\a:\r) arc (7*\a:11*\a:\r);
	    \draw[\style] (12*\a:\r) arc (12*\a:12.01*\a:\r);
	    \draw[\style] (13*\a:\r) arc (13*\a:17*\a:\r);
	 }

    \foreach \i in {0, 1, 2, 3, 5, 6, 7, 8, 10, 11, 12, 13, 14}
    {
        \node[vertex] (\i) at (\i*\a:\r) {};
    }
    
    \node at (4*\a:\r) {\dots};
    \node[rotate=90+9*\a] at (9*\a:\r) {\dots};
    \node[rotate=90+15*\a] at (15*\a:\r) {\dots};

    \node[label={[shift={(0.3,-0.25)}]\scriptsize $i$}] at (12*\a:\r) {};
    \node[label={[shift={(0,-0.2)}]\scriptsize $i\hspace{-2pt}+\hspace{-2pt}q\hspace{-2pt}+\hspace{-2pt}2$}] at (3*\a:\l) {};
    \node[label={[shift={(0,-0.2)}]\scriptsize $i\hspace{-2pt}-\hspace{-2pt}q\hspace{-2pt}-\hspace{-2pt}2$}] at (5*\a:\l) {};

    \draw [decorate,decoration={brace,amplitude=5pt}]
        (1.9, 0.8) -- (1.9, -1.7) node[midway,xshift=12pt]{\scriptsize $Q_1$};
    \draw [decorate,decoration={brace,amplitude=5pt}]
        ($(6*\a:\l) + (0,0.25) + (4*\a:0.5cm)$) -- ($(2*\a:\l) + (0,0.25) + (4*\a:0.5cm)$) node[midway,yshift=10pt]{\scriptsize $P_2$};
    \draw [decorate,decoration={brace,amplitude=5pt}]
        (-1.9, -1.7) -- (-1.9, 0.8) node[midway,xshift=-12pt]{\scriptsize $Q_2$};
    \draw [decorate,decoration={brace,amplitude=5pt}]
        (12.3*\a:\l) -- (11.7*\a:\l) node[midway,yshift=-10pt]{\scriptsize $P_1$};
        
    \draw (3) -- (12);
    \draw (5) -- (12);
\end{tikzpicture}
        \hspace{1cm}
        \begin{tikzpicture}[baseline=0]
    \def \r{1.7cm}
    \def \rs{0.05*\r}
    \def \l{2cm}
    \def \b{0.3cm}
    \def\n{20}
    \pgfmathtruncatemacro{\nn}{\n-1}
    \def \a{360/\n}

    \draw[arealight2] plot [smooth cycle, tension=0.2] coordinates {
		(3*\a:\r) 
		(7*\a:\r) 
		(10*\a:\rs)
		(13*\a:\r) 
		(17*\a:\r)
		(20*\a:\rs)
    };

    \foreach \style in {area3, area2}
    {
	    \draw[\style] (3*\a:\r) arc (3*\a:7*\a:\r);
	    \draw[\style] (8*\a:\r) arc (8*\a:12*\a:\r);
	    \draw[\style] (13*\a:\r) arc (13*\a:17*\a:\r);
	    \draw[\style] (18*\a:\r) arc (18*\a:22*\a:\r);
	 }

    \foreach \i in {1, 2, 3, 4, 6, 7, 8, 9, 11, 12, 13, 14, 16, 17, 18, 19}
    {
        \node[vertex] (\i) at (\i*\a:\r) {};
    }

    \node[label={[shift={(-0.3,-0.45)}]\scriptsize $i$}] at (13*\a:\r) {};
    \node[label={[shift={(0.6,-0.55)}]\scriptsize$i\hspace{-2pt}+\hspace{-2pt}p\hspace{-2pt}-\hspace{-2pt}1$}] at (17*\a:\r) {};
    \node[label={[shift={(0.25,0.1)}]\scriptsize$i\hspace{-2pt}+\hspace{-2pt}p\hspace{-2pt}+\hspace{-2pt}q$}] at (3*\a:\r) {};
    \node[label={[shift={(-0.25,0.1)}]\scriptsize$i\hspace{-2pt}-\hspace{-2pt}q\hspace{-2pt}-\hspace{-2pt}1$}] at (7*\a:\r) {};

    \node[rotate=90] at (0*\a:\r) {\dots};
    \node[rotate=180] at (5*\a:\r) {\dots};
    \node[rotate=270] at (10*\a:\r) {\dots};
    \node[rotate=0] at (15*\a:\r) {\dots};

    \draw [decorate,decoration={brace,amplitude=5pt}]
        ($(2*\a:\l) + (\b, 0)$) -- ($(-2*\a:\l) + (\b, 0) $)node[midway,xshift=12pt]{\scriptsize $Q_1$};
    \draw [decorate,decoration={brace,amplitude=5pt}]
        ($(7*\a:\l) + (0,0.15) + (0, \b)$) -- ($(3*\a:\l) + (0,0.15) + (0, \b) $)node[midway,yshift=10pt]{\scriptsize $P_2$};
    \draw [decorate,decoration={brace,amplitude=5pt}]
        ($(12*\a:\l) + (-\b, 0)$) -- ($(8*\a:\l) + (-\b, 0) $)node[midway,xshift=-12pt]{\scriptsize $Q_2$};
    \draw [decorate,decoration={brace,amplitude=5pt}]
        ($(17*\a:\l) + (0, -\b)$) -- ($(13*\a:\l) + (0, -\b) $)node[midway,yshift=-10pt]{\scriptsize $P_1$};
        
    \draw (3) -- (16);
    \draw (4) -- (17);
    \draw (6) -- (13);
    \draw (7) -- (14);
\end{tikzpicture}
        \caption{Illustration of the proof that \eqref{eq:second-condition-equality} holds in case $|P_1| = p = 1$ (left) and $|P_1| = p \geq 2$ (right).
        In this and all subsequent figures, nodes are enumerated counterclockwise.
        The black lines depict those $q+2$-chords that are connected in the given partitions.}
        \label{fig:necessity-2}
    \end{figure}

    \textbf{Sufficiency.}
    Next, we show that the conditions are sufficient.
    That is, we prove that \eqref{eq:q-chored-k-cycle} induces a facet of the clique partitioning problem if conditions \ref{item:condition-1} and \ref{item:condition-2} are satisfied.
    The face induced by \eqref{eq:q-chored-k-cycle} is contained in a facet induced by some inequality $b^\top x \leq \beta$, i.e.~$\{ x \in \cpp_k \mid \eqref{eq:q-chored-k-cycle} \text{ satisfied at equality} \} \subseteq \set{x \in \cpp_k \mid b^\top x = \beta}$.
    We prove sufficiency by showing that $b^\top x \leq \beta$ and \eqref{eq:q-chored-k-cycle} differ only by a multiplicative constant $\lambda \in \mathbb{R}$, implying the corresponding faces to coincide and \eqref{eq:q-chored-k-cycle} to be facet-inducing.
    In particular, we show that there exists a $\lambda \in \mathbb{R}$ such that it holds for all $i \in \mathbb{Z}_k$ that $b_\set{i,i+1} = \lambda$, $b_\set{i,i+q} = -\lambda$ and $b_\set{i,i+p} = 0$ for all $p \in \{2,\dots,\lfloor k / 2\rfloor\} \setminus \{q\}$.
    Since there are vectors that satisfy both inequalities at equality, it then follows immediately that also $\beta = \lambda \left(k - \left\lceil \tfrac{k}{q} \right\rceil \right)$.

    To determine the coefficients $b$, we first introduce some notation that allows to concisely specify feasible solutions satisfying \eqref{eq:q-chored-k-cycle} at equality.
    Let $L = \lceil k / q \rceil$.
    An $L$-tuple $a \in \{1,\dots,q\}^{\Z_L}$ is called \emph{$kq$-feasible} if $\sum_{l \in \Z_L} a_l = k$ and $a_l \leq q$ for all $l \in \Z_L$.
    Any $i \in \Z_k$ together with a $kq$-feasible $a$ defines a cycle partition that, starting at node $i$, contains sets whose size is given by the elements of $a$. 
    We denote by $\varphi(i,a) = x^{\Pi^{\cyc}}$ the feasible vector induced by this cycle partition $\Pi^{\cyc} = \{\{i,\dots,i+a_0-1\}, \{i+a_0,\dots,i+a_0+a_1-1\},\dots, \{i-a_{L-1},\dots,i-1\}\}$. 
    Furthermore, we may overline two elements of a \emph{$kq$-feasible} $a$ to indicate that the associated sets in the cycle partition are subsets of the same set in the inducing partition, i.e.~for $a = (a_0, \dots, \overline{a_{j_1}}, \dots, \overline{a_{j_2}}, \dots, a_{L-1})$ we let $\varphi(i,a)$ denote the feasible vector of the partition $\Pi$ obtained from $\Pi^{\cyc}$ by joining the sets associated with $a_{j_1}$ and $a_{j_2}$. 
    By definition of $kq$-feasibility, $\Pi^{\cyc}$ satisfies \ref{item:eq-a} of \Cref{lem:part}, and $\Pi$ satisfies \ref{item:eq-a} of \Cref{lem:part} if there are at least $q$ nodes between the components of the cycle partition corresponding to the overlined elements.
    Thus, $\varphi(i,a)$ satisfies \eqref{eq:q-chored-k-cycle} at equality in this case.
    To give an example of the introduced notation, the first partition in \Cref{fig:cycle-partition} (shaded area) can be expressed by $\varphi(i, (\overline{3},3,\overline{2},3,2))$ and the corresponding cycle partition (darker shade) by $\varphi(i, (3,3,2,3,2))$, where $i$ is the node at the bottom.

    \begin{figure}[t]
        \centering
        \begin{tikzpicture}
    \def \r{1.2cm}
    \def \l{0.85cm}
    \def \s{20}
    \def \d{1}

    \foreach \style in {area3, area2}
    {
        \draw[\style] (-90-3*\s+\d:\r) arc (-90-3*\s+\d:-90-1*\s-\d:\r);
        \draw[\style] (-90+\d:\r) arc (-90+\d:-90+4*\s-\d:\r);
        \draw[\style] (-90+5*\s+\d:\r) arc (-90+5*\s+\d:-90+7*\s-\d:\r);
        \draw[\style] (-90+8*\s+\d:\r) arc (-90+8*\s+\d:-90+10*\s-\d:\r);
        \draw[\style] (-90-6*\s+\d:\r) arc (-90-6*\s+\d:-90-4*\s-\d:\r);
	 }

    \node[rotate=-2.2*\s] at (-90-2*\s:\l) {\scriptsize $q\hspace{-2pt}-\hspace{-2pt}1$};
    \node[rotate=1.8*\s] at (-90+2*\s:\l) {\scriptsize $p\hspace{-2pt}+\hspace{-2pt}1$};
    \node[rotate=180+6*\s] at (-90+6*\s:\l) {\scriptsize $q\hspace{-2pt}-\hspace{-2pt}p\hspace{-2pt}+\hspace{-2pt}1$};
    \node at (-90+9*\s:\l) {\scriptsize $q$};
    \node at (-90-5*\s:\l) {\scriptsize $q$};

    \node[vertex, label={[shift={(0.0,-0.6)}]\scriptsize $i$}] at (-90:\r) {};
    \node[vertex] at (-90+\s:\r) {};
    \node[rotate=2*\s] at (-90+2*\s:\r) {$\cdots$};
    \node[vertex] at (-90+3*\s:\r) {};
    \node[vertex] at (-90+4*\s:\r) {};
    \node[vertex] at (-90+5*\s:\r) {};
    \node[rotate=6*\s] at (-90+6*\s:\r) {$\cdots$};
    \node[vertex] at (-90+7*\s:\r) {};
    \node[vertex] at (-90+8*\s:\r) {};
    \node[rotate=9*\s] at (-90+9*\s:\r) {$\cdots$};
    \node[vertex] at (-90+10*\s:\r) {};
    \node[vertex] at (-90-\s:\r) {};
    \node[rotate=-2*\s] at (-90-2*\s:\r) {$\cdots$};
    \node[vertex] at (-90-3*\s:\r) {};
    \node[vertex] at (-90-4*\s:\r) {};
    \node[rotate=-5*\s] at (-90-5*\s:\r) {$\cdots$};
    \node[vertex] at (-90-6*\s:\r) {};

    \node[rotate=11*\s] at (-90+11*\s:\r) {$\cdots$};
\end{tikzpicture}%
        \hfill
        \begin{tikzpicture}
    \def \r{1.2cm}
    \def \l{0.85cm}
    \def \s{20}
    \def \d{1}

    \foreach \style in {area3, area2}
    {
        \draw[\style] (-90-3*\s+\d:\r) arc (-90-3*\s+\d:-90-\s-\d:\r);
        \draw[\style] (-90+\d:\r) arc (-90+\d:-90+3*\s-\d:\r);
        \draw[\style] (-90+4*\s+\d:\r) arc (-90+4*\s+\d:-90+7*\s-\d:\r);
        \draw[\style] (-90+8*\s+\d:\r) arc (-90+8*\s+\d:-90+10*\s-\d:\r);
        \draw[\style] (-90+12*\s+\d:\r) arc (-90+12*\s+\d:-90+14*\s-\d:\r);
	 }

    \node[rotate=-2*\s] at (-90-2*\s:\l) {\scriptsize $q\hspace{-2pt}-\hspace{-2pt}1$};
    \node at (-90+1.5*\s:\l) {\scriptsize $p$};
    \node[rotate=-3.5*\s] at (-90+5.5*\s:\l) {\scriptsize $q\hspace{-2pt}-\hspace{-2pt}p\hspace{-2pt}+\hspace{-2pt}2$};
    \node at (-90+9*\s:\l) {\scriptsize $q$};
    \node at (-90-5*\s:\l) {\scriptsize $q$};

    \node[vertex, label={[shift={(0.0,-0.6)}]\scriptsize $i$}] at (-90:\r) {};
    \node[vertex] at (-90+\s:\r) {};
    \node[rotate=2*\s] at (-90+2*\s:\r) {$\cdots$};
    \node[vertex] at (-90+3*\s:\r) {};
    \node[vertex] at (-90+4*\s:\r) {};
    \node[vertex] at (-90+5*\s:\r) {};
    \node[rotate=6*\s] at (-90+6*\s:\r) {$\cdots$};
    \node[vertex] at (-90+7*\s:\r) {};
    \node[vertex] at (-90+8*\s:\r) {};
    \node[rotate=9*\s] at (-90+9*\s:\r) {$\cdots$};
    \node[vertex] at (-90+10*\s:\r) {};
    \node[vertex] at (-90-\s:\r) {};
    \node[rotate=-2.2*\s] at (-90-2*\s:\r) {$\cdots$};
    \node[vertex] at (-90-3*\s:\r) {};
    \node[vertex] at (-90-4*\s:\r) {};
    \node[rotate=-5*\s] at (-90-5*\s:\r) {$\cdots$};
    \node[vertex] at (-90-6*\s:\r) {};

    \node[rotate=11*\s] at (-90+11*\s:\r) {$\cdots$};
\end{tikzpicture}%
        \hfill
        \begin{tikzpicture}
    \def \r{1.2cm}
    \def \l{0.85cm}
    \def \s{20}
    \def \d{1}

    \foreach \style in {area3, area2}
    {
        \draw[\style] (-90-3*\s+\d:\r) arc (-90-3*\s+\d:-90-\d:\r);
        \draw[\style] (-90+\s+\d:\r) arc (-90+\s+\d:-90+4*\s-\d:\r);
        \draw[\style] (-90+5*\s+\d:\r) arc (-90+5*\s+\d:-90+7*\s-\d:\r);
        \draw[\style] (-90+8*\s+\d:\r) arc (-90+8*\s+\d:-90+10*\s-\d:\r);
        \draw[\style] (-90-6*\s+\d:\r) arc (-90-6*\s+\d:-90-4*\s-\d:\r);
	 }

    \draw[area2] (-90:\r) arc (-90:-90-3*\s:\r);
    \node at (-90-1.5*\s:\l) {\scriptsize $q$};
    \draw[area2] (-90+\s:\r) arc (-90+\s:-90+4*\s:\r);
    \node at (-90+2.5*\s:\l) {\scriptsize $p$};
    \draw[area2] (-90+5*\s:\r) arc (-90+5*\s:-90+7*\s:\r);
    \node[rotate=-3*\s] at (-90+6*\s:\l) {\scriptsize $q\hspace{-2pt}-\hspace{-2pt}p\hspace{-2pt}+\hspace{-2pt}1$};
    \draw[area2] (-90+8*\s:\r) arc (-90+8*\s:-90+10*\s:\r);
    \node at (-90+9*\s:\l) {\scriptsize $q$};
    \draw[area2] (-90-4*\s:\r) arc (-90-4*\s:-90-6*\s:\r);
    \node at (-90-5*\s:\l) {\scriptsize $q$};

    \node[vertex, label={[shift={(0.0,-0.6)}]\scriptsize $i$}] at (-90:\r) {};
    \node[vertex] at (-90+\s:\r) {};
    \node[rotate=2*\s] at (-90+2*\s:\r) {$\cdots$};
    \node[vertex] at (-90+3*\s:\r) {};
    \node[vertex] at (-90+4*\s:\r) {};
    \node[vertex] at (-90+5*\s:\r) {};
    \node[rotate=6*\s] at (-90+6*\s:\r) {$\cdots$};
    \node[vertex] at (-90+7*\s:\r) {};
    \node[vertex] at (-90+8*\s:\r) {};
    \node[rotate=9*\s] at (-90+9*\s:\r) {$\cdots$};
    \node[vertex] at (-90+10*\s:\r) {};
    \node[vertex] at (-90-\s:\r) {};
    \node[rotate=-2*\s] at (-90-2*\s:\r) {$\cdots$};
    \node[vertex] at (-90-3*\s:\r) {};
    \node[vertex] at (-90-4*\s:\r) {};
    \node[rotate=-5*\s] at (-90-5*\s:\r) {$\cdots$};
    \node[vertex] at (-90-6*\s:\r) {};

    \node[rotate=11*\s] at (-90+11*\s:\r) {$\cdots$};
\end{tikzpicture}%
        \hfill
        \begin{tikzpicture}
    \def \r{1.2cm}
    \def \l{0.85cm}
    \def \s{20}
    \def \d{1}

    \foreach \style in {area3, area2}
    {
        \draw[\style] (-90-3*\s+\d:\r) arc (-90-3*\s+\d:-90-0*\s-\d:\r);
        \draw[\style] (-90+\s+\d:\r) arc (-90+\s+\d:-90+3*\s-\d:\r);
        \draw[\style] (-90+4*\s+\d:\r) arc (-90+4*\s+\d:-90+7*\s-\d:\r);
        \draw[\style] (-90+8*\s+\d:\r) arc (-90+8*\s+\d:-90+10*\s-\d:\r);
        \draw[\style] (-90-6*\s+\d:\r) arc (-90-6*\s+\d:-90-4*\s-\d:\r);
	 }

    \node at (-90-1.5*\s:\l) {\scriptsize $q$};
    \node[rotate=1.8*\s] at (-90+2*\s:\l) {\scriptsize $p\hspace{-2pt}-\hspace{-2pt}1$};
    \node[rotate=180+5.5*\s] at (-90+5.5*\s:\l) {\scriptsize $q\hspace{-2pt}-\hspace{-2pt}p\hspace{-2pt}+\hspace{-2pt}2$};
    \node at (-90+9*\s:\l) {\scriptsize $q$};
    \node at (-90-5*\s:\l) {\scriptsize $q$};

    \node[vertex, label={[shift={(0.0,-0.6)}]\scriptsize $i$}] at (-90:\r) {};
    \node[vertex] at (-90+\s:\r) {};
    \node[rotate=2*\s] at (-90+2*\s:\r) {$\cdots$};
    \node[vertex] at (-90+3*\s:\r) {};
    \node[vertex] at (-90+4*\s:\r) {};
    \node[vertex] at (-90+5*\s:\r) {};
    \node[rotate=6*\s] at (-90+6*\s:\r) {$\cdots$};
    \node[vertex] at (-90+7*\s:\r) {};
    \node[vertex] at (-90+8*\s:\r) {};
    \node[rotate=9*\s] at (-90+9*\s:\r) {$\cdots$};
    \node[vertex] at (-90+10*\s:\r) {};
    \node[vertex] at (-90-\s:\r) {};
    \node[rotate=-2*\s] at (-90-2*\s:\r) {$\cdots$};
    \node[vertex] at (-90-3*\s:\r) {};
    \node[vertex] at (-90-4*\s:\r) {};
    \node[rotate=-5*\s] at (-90-5*\s:\r) {$\cdots$};
    \node[vertex] at (-90-6*\s:\r) {};

    \node[rotate=11*\s] at (-90+11*\s:\r) {$\cdots$};
\end{tikzpicture}%
        \caption{Depicted are four partitions (with their corresponding cycle partitions) that induce the feasible solutions on the right-hand side of \eqref{eq:construction-short-chords}. 
        In this and all subsequent figures, the size of individual components of cycle partitions is indicated by numbers within the circle.
        To verify that \eqref{eq:construction-short-chords} holds, observe that the edge $\{i, i+p\}$ is only connected in the first partition while all other edges are connected in either two or all four partitions.
        }
        \label{fig:short-chord-unit-vectors}
    \end{figure}

    \textbf{Coefficients $\bm{b_\set{i,i+p} = 0}$ for $\bm{p \in \set{2,\dots,q-1} }$ and $\bm{i \in \Z_k}$:}
    Let $i \in \Z_k$ and let $p \in \{2,\dots,q-1\}$.
    We can construct $b_{\{i,i+p\}}$ using a combination of feasible solutions defined by $kq$-feasible $L$-tuples:
    \begin{equation}
    \begin{alignedat}{5} \label{eq:construction-short-chords}
        b_{\{i,i+p\}} 
        & = b^\top \varphi(i,\     &&(p+1, \ &&q-p+1, \ q, \dots, q, \ &&q-1&&))            \\
        & - b^\top \varphi(i,\     &&(p, \ &&q-p+2, \ q, \dots, q, \ &&q-1&&))              \\
        & - b^\top \varphi(i + 1,\ &&(p, \ &&q-p+1, \ q, \dots, q, \ &&q&&))                \\
        & + b^\top \varphi(i + 1,\ &&(p-1, \ &&q-p+2, \ q, \dots, q, \ &&q&&))\enspace .
    \end{alignedat}
    \end{equation}
    By the discussion above, these feasible solutions satisfy \eqref{eq:q-chored-k-cycle}, and thus also $b^\top x \leq \beta$, at equality.
    Consequently, \eqref{eq:construction-short-chords} equals $\beta - \beta - \beta + \beta = 0$, yielding the desired result of $b_{\{i,i+p\}} = 0$.
    The feasible solutions used in this construction are illustrated in \Cref{fig:short-chord-unit-vectors}.
    
    \textbf{Coefficients $\bm{b_{\{i,i+1\}} = \lambda}$ and $\bm{b_{\{i,i+q\}} = -\lambda}$ for $\bm{i \in \Z_k}$:}    
    To show that there exists a $\lambda \in \mathbb{R}$ such that $b_{\{i,i+1\}} = \lambda$ and $b_{\{i,i+q\}} = -\lambda$ for all $i \in \Z_k$, we first show $b_{\{i,i+1\}} = b_{\{i+1,i+2\}}$, and then $b_{\{i,i+1\}} = -b_{\{i,i+q\}}$.
    The feasible solutions used for deriving these equalities are illustrated in \Cref{fig:edges-and-q-chords}.
    Note that $b_{\{i,i+p\}} = 0$ for $p \in \{2,\dots,q-1\}$ by the previous case, and these coefficients must not be considered in the following.
    
    It holds
    \begin{align*}
        b_{\{i,i+1\}} - b_{\{i+1,i+2\}} = b^\top \varphi(i, (2, q-1, q,\dots,q)) - b^\top \varphi(i, (1, q, q,\dots,q))
    \end{align*}
    where the right-hand side evaluates to $\beta - \beta = 0$ since both feasible solutions satisfy $b^\top x \leq \beta$ at equality, yielding $b_{\{i,i+1\}} = b_{\{i+1,i+2\}}$.

    Let now
    \begin{align}\label{eq:q-chord-partition}
        \Pi = \{\{i,\dots,i+q\}, \{i+q+1,\dots,i+2q\},\dots, \{i-q,\dots,i-1\}\} 
    \end{align} 
    be the partition where $i$ is the first node of a component of size $q+1$ and all other components have size $q$.
    It holds
    \begin{align*}\label{eq:edge-q-chord-pairs}
        b_{\{i,i+1\}} + b_{\{i,i+q\}} = b^\top x^\Pi - b^\top \varphi(i, (1, q, \dots, q)) \enspace .
    \end{align*}
    Again, $\varphi(i, (1, q, \dots, q))$ satisfies $b^\top x \leq \beta$ at equality.
    Further, $\Pi$ satisfies \ref{item:eq-b} of \Cref{lem:part}, and thus also \eqref{eq:q-chored-k-cycle} and $b^\top x \leq \beta$ at equality.
    Consequently, the right-hand side above evaluates to $\beta - \beta = 0$, yielding $b_{\{i,i+1\}} = - b_{\{i,i+q\}}$.

    \begin{figure}[t]
        \centering
        \begin{tikzpicture}
    \def \r{0.9cm}
    \def \l{0.5cm}
    \def \s{30}

    \foreach \style in {area3, area2}
    {
        \draw[\style] (-90:\r) arc (-90:-90+1*\s:\r);
        \draw[\style] (-90+2*\s:\r) arc (-90+2*\s:-90+4*\s:\r);
        \draw[\style] (-90+5*\s:\r) arc (-90+5*\s:-90+7*\s:\r);
        \draw[\style] (-90+9*\s:\r) arc (-90+9*\s:-90+11*\s:\r);
	 }

    \node[rotate=2.8*\s] at (-90+3*\s:\l) {\scriptsize $q\hspace{-2pt}-\hspace{-2pt}1$};
    \node at (-90+6*\s:\l) {\scriptsize $q$};
    \node at (-90+10*\s:\l) {\scriptsize $q$};
    
    \node[vertex, label={[shift={(0.0,-0.6)}]\scriptsize $i$}] at (-90:\r) {};
    \node[vertex] at (-90+\s:\r) {};
    \node[vertex] at (-90+2*\s:\r) {};
    \node[rotate=3*\s] at (-90+3*\s:\r) {$\cdots$};
    \node[vertex] at (-90+4*\s:\r) {};
    \node[vertex] at (-90+5*\s:\r) {};
    \node[rotate=6*\s] at (-90+6*\s:\r) {$\cdots$};
    \node[vertex] at (-90+7*\s:\r) {};
    \node[rotate=8*\s] at (-90+8*\s:\r) {$\cdots$};
    \node[vertex] at (-90+9*\s:\r) {};
    \node[rotate=10*\s] at (-90+10*\s:\r) {$\cdots$};
    \node[vertex] at (-90+11*\s:\r) {};
\end{tikzpicture}%
        \hfill
        \begin{tikzpicture}
    \def \r{0.9cm}
    \def \l{0.5cm}
    \def \s{30}

    \foreach \style in {area3, area2}
    {
        \draw[\style] (-90:\r) arc (-90:-90:\r);
        \draw[\style] (-90+1*\s:\r) arc (-90+1*\s:-90+4*\s:\r);
        \draw[\style] (-90+5*\s:\r) arc (-90+5*\s:-90+7*\s:\r);
        \draw[\style] (-90+9*\s:\r) arc (-90+9*\s:-90+11*\s:\r);
	 }

    \node at (-90+2.5*\s:\l) {\scriptsize $q$};
    \node at (-90+6*\s:\l) {\scriptsize $q$};
    \node at (-90+10*\s:\l) {\scriptsize $q$};
    
    \node[vertex, label={[shift={(0.0,-0.6)}]\scriptsize $i$}] at (-90:\r) {};
    \node[vertex] at (-90+\s:\r) {};
    \node[vertex] at (-90+2*\s:\r) {};
    \node[rotate=3*\s] at (-90+3*\s:\r) {$\cdots$};
    \node[vertex] at (-90+4*\s:\r) {};
    \node[vertex] at (-90+5*\s:\r) {};
    \node[rotate=6*\s] at (-90+6*\s:\r) {$\cdots$};
    \node[vertex] at (-90+7*\s:\r) {};
    \node[rotate=8*\s] at (-90+8*\s:\r) {$\cdots$};
    \node[vertex] at (-90+9*\s:\r) {};
    \node[rotate=10*\s] at (-90+10*\s:\r) {$\cdots$};
    \node[vertex] at (-90+11*\s:\r) {};
\end{tikzpicture}%
        \hfill
        \begin{tikzpicture}
    \def \r{0.9cm}
    \def \l{0.5cm}
    \def \s{30}

    \foreach \style in {area3, area2}
    {
        \draw[\style] (-90:\r) arc (-90:-90+4*\s:\r);
        \draw[\style] (-90+5*\s:\r) arc (-90+5*\s:-90+7*\s:\r);
        \draw[\style] (-90+9*\s:\r) arc (-90+9*\s:-90+11*\s:\r);
	 }

    \node[rotate=1.8*\s] at (-90+2*\s:\l) {\scriptsize $q\hspace{-2pt}+\hspace{-2pt}1$};
    \node at (-90+6*\s:\l) {\scriptsize $q$};
    \node at (-90+10*\s:\l) {\scriptsize $q$};

    \node[vertex, label={[shift={(0.0,-0.6)}]\scriptsize $i$}] at (-90:\r) {};
    \node[vertex] at (-90+\s:\r) {};
    \node[vertex] at (-90+2*\s:\r) {};
    \node[rotate=3*\s] at (-90+3*\s:\r) {$\cdots$};
    \node[vertex] at (-90+4*\s:\r) {};
    \node[vertex] at (-90+5*\s:\r) {};
    \node[rotate=6*\s] at (-90+6*\s:\r) {$\cdots$};
    \node[vertex] at (-90+7*\s:\r) {};
    \node[rotate=8*\s] at (-90+8*\s:\r) {$\cdots$};
    \node[vertex] at (-90+9*\s:\r) {};
    \node[rotate=10*\s] at (-90+10*\s:\r) {$\cdots$};
    \node[vertex] at (-90+11*\s:\r) {};
\end{tikzpicture}%
        \hfill
        \begin{tikzpicture}
    \def \r{0.9cm}
    \def \l{0.6cm}
    \def \s{30}

    \node[vertex, label={[shift={(0.0,-0.6)}]\scriptsize $i$}] (i) at (-90:\r) {};
    \node[vertex] (i+1) at (-90+\s:\r) {};
    \node[vertex] (i+2) at (-90+2*\s:\r) {};
    \node[rotate=3*\s] at (-90+3*\s:\r) {$\cdots$};
    \node[vertex] at (-90+4*\s:\r) {};
    \node[vertex] at (-90+5*\s:\r) {};
    \node[rotate=6*\s] at (-90+6*\s:\r) {$\cdots$};
    \node[vertex] at (-90+7*\s:\r) {};
    \node[rotate=8*\s] at (-90+8*\s:\r) {$\cdots$};
    \node[vertex] at (-90+9*\s:\r) {};
    \node[rotate=10*\s] at (-90+10*\s:\r) {$\cdots$};
    \node[vertex] at (-90+11*\s:\r) {};

    \draw (i) -- (i+1);
    \draw (i+1) -- (i+2);
\end{tikzpicture}%
        \hfill
        \begin{tikzpicture}
    \def \r{0.9cm}
    \def \l{0.6cm}
    \def \s{30}

    \node[vertex, label={[shift={(0.0,-0.6)}]\scriptsize $i$}] (i) at (-90:\r) {};
    \node[vertex] (i+1) at (-90+\s:\r) {};
    \node[vertex] at (-90+2*\s:\r) {};
    \node[rotate=3*\s] at (-90+3*\s:\r) {$\cdots$};
    \node[vertex] (i+q) at (-90+4*\s:\r) {};
    \node[vertex] at (-90+5*\s:\r) {};
    \node[rotate=6*\s] at (-90+6*\s:\r) {$\cdots$};
    \node[vertex] at (-90+7*\s:\r) {};
    \node[rotate=8*\s] at (-90+8*\s:\r) {$\cdots$};
    \node[vertex] at (-90+9*\s:\r) {};
    \node[rotate=10*\s] at (-90+10*\s:\r) {$\cdots$};
    \node[vertex] at (-90+11*\s:\r) {};

    \draw (i) -- (i+1);
    \draw (i) -- (i+q);
\end{tikzpicture}%
        \caption{
            Depicted from left to right are: The partitions associated with the feasible solutions $\varphi(i, (2, q-1, q, \dots, q))$ and $\varphi(i, (1, q, \dots, q))$, the partition $\Pi$ of \eqref{eq:q-chord-partition}, the edges with coefficients $b_{\{i,i+1\}}$, $b_{\{i+1,i+2\}}$, and with coefficients $b_{\{i,i+1\}}$, $b_{\{i,i+q\}}$.
        }
        \label{fig:edges-and-q-chords}
    \end{figure}

    \textbf{Coefficients $\bm{b_{\{i,i+p\}} = 0}$ for $\bm{p \in \set{q+1,\dots,\lfloor \tfrac{k}{2} \rfloor} }$ and $\bm{i \in \Z_k}$:}
    In the following we will frequently consider the difference of the feasible solutions induced by a partition and its corresponding cycle partition.
    The resulting vector then only contains non-zero entries for chords longer than $q$.
    We introduce the following notation for this construction:
    For any $kq$-feasible $L$-tuple $a = (a_0, \dots, \overline{a_{j_1}}, \dots, \overline{a_{j_2}}, \dots, a_{L-1})$ and any $i \in \Z_k$, define $\psi(i, a) = \varphi(i, a) - \varphi(i, (a_0, \dots, a_{L-1}))$.
    Note that $b^\top \psi(i, a) = \beta - \beta = 0$.

    By Condition \ref{item:condition-1}, $k = mq+1$ for some $m \in \mathbb{N}$. 
    We make a case distinction on $m$, where the case $m<3$ does not need to be considered as $\{q+1,\dots,\lfloor \tfrac{k}{2} \rfloor\}$ would be empty.

    To begin with, let $m=3$, i.e.~$k = 3q+1$.
    The case of $q=3$ is easily verified as facet-inducing.
    Thus, let $q > 3$. Now, by Condition \ref{item:condition-2}, $q$ must be even.
    We show that in this case $b_\set{i,i+p} = 0$ for $i \in \Z_k$ by induction over $p = q+1,\dots,\lfloor \tfrac{k}{2} \rfloor$.
    In particular, we show that $b_\set{j,j+p} = -b_\set{j-q, j+p-q}$ for all $j \in \Z_k$.
    Applying this equality $k$ times to substitute the right-hand side yields $b_\set{j,j+p} = (-1)^k b_\set{j-kq, j+p-kq}$.
    Taking addition modulo $k$ into account, this simplifies to $b_\set{j,j+p} = (-1)^k b_\set{j, j+p}$.
    Finally, as $q$ is even, $k$ must be odd, implying $b_\set{j,j+p} = -b_\set{j,j+p}$, and thus the desired result of $b_\set{j,j+p} = 0$.
    
    Assume that $b_\set{j,j+p'} = 0$ for all $j \in \Z_k$, $p' \in \{q+1,\dots,p-1\}$. 
    In particular, this assumption holds for the induction start $p = q+1$.
    It holds
    \begin{equation}
    \begin{alignedat}{5} \label{eq:construction-n3}
        \sum_{p'=q+1}^{p} \rlap{$\left(b_\set{j+p-p',j+p} + b_\set{j+p-q-p',j+p-q}\right)$} \\
        &= b^\top \psi(j+(p-q),\ &&(\overline{1},\ &&q,\ &&\overline{q},\ &&q))\\ 
        &+ b^\top \psi(j,\ &&(\overline{(p-q)},\ &&q,\ &&\overline{q+1-(p-q)},\ &&q)) \\ 
        &- b^\top \psi(j,\ &&(\overline{(p-q)+1},\ &&q,\ &&\overline{q-(p-q)},\ &&q)) \enspace ,
    \end{alignedat}
    \end{equation}
    as illustrated in \Cref{fig:n3-unit-vectors}.
    Since all coefficients in the sum for $p'=q+1,\dots,p-1$ are $0$ by the induction hypothesis, and the right-hand side evaluates to $0$ by the discussion above, the equation simplifies to $b_\set{j,j+p} + b_\set{j-q, j+p-q} = 0$ which implies the desired $b_\set{j,j+p} = - b_\set{j-q, j+p-q}$.
    This completes the induction and thus the case of $k = 3q+1$.

    For finishing our case distinction, it remains to regard the case of $k = mq+1$ with $m>3$.
    Let $i \in \Z_k$ and $s \in \set{1,\dots,m-3}$. 
    We firstly show $b_\set{i,i+sq+j} = 0$ for $j \in \set{2, \dots, q-1}$, then $b_\set{i,i+sq+q} = 0$, and lastly $b_\set{i,i+sq+1} = 0$.
    By doing so, we are showing that all coefficients $b_\set{i,i+l}$ with $l \in \{q+1,\dots,k-2q-1\}$ are $0$ which, as $m>3$, correspond exactly to the ones of the $p$-chords with $p \in \{q+1,\dots,\lfloor k/2\rfloor\}$.
    Note that, due to the addition modulo $k$, we are in fact discussing some coefficients $b_\set{i, i+l}$ twice, but avoid further case distinctions.

    \begin{figure}[t]
        \centering
        \hspace{0.2cm}
        \begin{tikzpicture}
    \def \r{1.25cm}
    \def \rs{0.05*\r}
    \def \l{0.9cm}
    \pgfmathtruncatemacro{\n}{17}
    \def \s{360/\n}
    \def \d{1}

    \draw[arealight2] plot [smooth cycle, tension=0.2] coordinates {
		(-90:\r)
        (-90:\r)  
		(-90+4*\s:\rs) 
		(-90+5*\s+\d:\r)
        (-90+6.5*\s:\r)  
		(-90+10*\s-\d:\r) 
		(-90+0*\s:\rs) 
    };

    \foreach \style in {area3, area2}
    {
        \draw[\style] (-90:\r) arc (-90:-90:\r);
        \draw[\style] (-90+1*\s+\d:\r) arc (-90+1*\s+\d:-90+4*\s-\d:\r);
        \draw[\style] (-90+5*\s+\d:\r) arc (-90+5*\s+\d:-90+10*\s-\d:\r);
        \draw[\style] (-90+11*\s+\d:\r) arc (-90+11*\s+\d:-90+16*\s-\d:\r);
	 }

    \node at (-90+2.5*\s:\l) {\scriptsize $q$};
    \node at (-90+7.5*\s:\l) {\scriptsize $q$};
    \node at (-90+13.5*\s:\l) {\scriptsize $q$};

    \node[vertex, label={[shift={(0.0,-0.75)}]\scriptsize $j\hspace{-2pt}+\hspace{-2pt}p\hspace{-2pt}-\hspace{-2pt}q$}] (j2) at (-90:\r) {};
    \node[vertex] at (-90+1*\s:\r) {};
    \node[rotate=2*\s] at (-90+2*\s:\r) {\tiny $\cdots$};
    \node[vertex] at (-90+3*\s:\r) {};
    \node[vertex, label={[shift={(0.55,-0.25)}]\scriptsize $j\hspace{-2pt}+\hspace{-2pt}p$}] (a4) at (-90+4*\s:\r) {};

    \node[vertex] at (-90+5*\s:\r) {};
    \node[rotate=6*\s] at (-90+6*\s:\r) {\tiny$\cdots$};
    \node[vertex] at (-90+7*\s:\r) {};
    \node[vertex,label={[shift={(0.0,0.1)}]\scriptsize $j\hspace{-2pt}-\hspace{-2pt}q$}] (a3) at (-90+8*\s:\r) {};

    \node[rotate=9*\s] (a2) at (-90+9*\s:\r) {\tiny$\cdots$};
    \node (a2-1) at (-90+9.2*\s:\r) {};
    \node (a2-2) at (-90+8.95*\s:\r) {};
    \node (a2-3) at (-90+8.7*\s:\r) {};

    \node[vertex] (a1) at (-90+10*\s:\r) {};

    \node[vertex] at (-90+11*\s:\r) {};
    \node[rotate=12*\s] at (-90+12*\s:\r) {\tiny$\cdots$};
    \node[vertex] at (-90+13*\s:\r) {};

    \node[vertex, label={[shift={(-0.3,-0.25)}]\scriptsize $j$}] (j1) at (-90+14*\s:\r) {};
    \node[rotate=15*\s] at (-90+15*\s:\r) {\tiny$\cdots$};
    \node[vertex] at (-90+16*\s:\r) {};
\end{tikzpicture}%
        \hfill
        \begin{tikzpicture}
    \def \r{1.25cm}
    \def \rs{0.05*\r}
    \def \l{0.9cm}
    \pgfmathtruncatemacro{\n}{17}
    \def \s{360/\n}
    \def \d{1}

    \draw[arealight2] plot [smooth cycle, tension=0.2] coordinates {
		(-90+14*\s+\d:\r)
        (-90+16*\s-\d:\r)  
		(-90+4*\s:\rs) 
		(-90+4*\s+\d:\r) 
		(-90+7*\s-\d:\r) 
		(-90+0*\s:\rs) 
    };

    \foreach \style in {area3, area2}
    {
        \draw[\style] (-90+\d:\r) arc (-90+\d:-90+3*\s-\d:\r);
        \draw[\style] (-90+4*\s+\d:\r) arc (-90+4*\s+\d:-90+7*\s-\d:\r);
        \draw[\style] (-90+8*\s+\d:\r) arc (-90+8*\s+\d:-90+13*\s-\d:\r);
        \draw[\style] (-90+14*\s+\d:\r) arc (-90+14*\s+\d:-90+16*\s-\d:\r);
	 }

    \node at (-90+1.5*\s:\l) {\scriptsize $q$};
    \node at (-90+10.5*\s:\l) {\scriptsize $q$};

    \node[vertex, label={[shift={(0.0,-0.75)}]\scriptsize $j\hspace{-2pt}+\hspace{-2pt}p\hspace{-2pt}-\hspace{-2pt}q$}] at (-90:\r) {};
    \node[vertex] at (-90+1*\s:\r) {};
    \node[rotate=2*\s] at (-90+2*\s:\r) {\tiny $\cdots$};
    \node[vertex] at (-90+3*\s:\r) {};
    \node[vertex, label={[shift={(0.55,-0.25)}]\scriptsize $j\hspace{-2pt}+\hspace{-2pt}p$}] at (-90+4*\s:\r) {};

    \node[vertex] at (-90+5*\s:\r) {};
    \node[rotate=6*\s] at (-90+6*\s:\r) {\tiny$\cdots$};
    \node[vertex] at (-90+7*\s:\r) {}; 
    \node[vertex,label={[shift={(0.0,0.1)}]\scriptsize $j\hspace{-2pt}-\hspace{-2pt}q$}] at (-90+8*\s:\r) {};
    \node[rotate=9*\s] at (-90+9*\s:\r) {\tiny$\cdots$};
    \node[vertex] at (-90+10*\s:\r) {};

    \node[vertex] at (-90+11*\s:\r) {};
    \node[rotate=12*\s] at (-90+12*\s:\r) {\tiny$\cdots$};
    \node[vertex] at (-90+13*\s:\r) {};

    \node[vertex, label={[shift={(-0.3,-0.25)}]\scriptsize $j$}] at (-90+14*\s:\r) {};
    \node[rotate=15*\s] at (-90+15*\s:\r) {\tiny$\cdots$};
    \node[vertex] at (-90+16*\s:\r) {};
\end{tikzpicture}%
        \hfill
        \begin{tikzpicture}
    \def \r{1.25cm}
    \def \rs{0.05*\r}
    \def \l{0.9cm}
    \pgfmathtruncatemacro{\n}{17}
    \def \s{360/\n}
    \def \d{1}

    \draw[arealight2] plot [smooth cycle, tension=0.2] coordinates {
		(-90+14*\s+\d:\r)
        (-90+17*\s-\d:\r)  
		(-90+4*\s:\rs) 
		(-90+5*\s+\d:\r) 
		(-90+7*\s-\d:\r) 
		(-90+13*\s:\rs) 
    };

    \foreach \style in {area3, area2}
    {
        \draw[\style] (-90+1*\s+\d:\r) arc (-90+1*\s+\d:-90+4*\s-\d:\r);
        \draw[\style] (-90+5*\s+\d:\r) arc (-90+5*\s+\d:-90+7*\s-\d:\r);
        \draw[\style] (-90+8*\s+\d:\r) arc (-90+8*\s+\d:-90+13*\s-\d:\r);
        \draw[\style] (-90+14*\s+\d:\r) arc (-90+14*\s+\d:-90+17*\s-\d:\r);
	}

    \node at (-90+2.5*\s:\l) {\scriptsize $q$};
    \node at (-90+10.5*\s:\l) {\scriptsize $q$};

    \node[vertex, label={[shift={(0.0,-0.75)}]\scriptsize $j\hspace{-2pt}+\hspace{-2pt}p\hspace{-2pt}-\hspace{-2pt}q$}] at (-90:\r) {};
    \node[vertex] at (-90+1*\s:\r) {};
    \node[rotate=2*\s] at (-90+2*\s:\r) {\tiny $\cdots$};
    \node[vertex] at (-90+3*\s:\r) {};
    \node[vertex, label={[shift={(0.55,-0.25)}]\scriptsize $j\hspace{-2pt}+\hspace{-2pt}p$}] at (-90+4*\s:\r) {};

    \node[vertex] at (-90+5*\s:\r) {};
    \node[rotate=6*\s] at (-90+6*\s:\r) {\tiny$\cdots$};
    \node[vertex] at (-90+7*\s:\r) {}; 
    \node[vertex,label={[shift={(0.0,0.1)}]\scriptsize $j\hspace{-2pt}-\hspace{-2pt}q$}] at (-90+8*\s:\r) {};
    \node[rotate=9*\s] at (-90+9*\s:\r) {\tiny$\cdots$};
    \node[vertex] at (-90+10*\s:\r) {};

    \node[vertex] at (-90+11*\s:\r) {};
    \node[rotate=12*\s] at (-90+12*\s:\r) {\tiny$\cdots$};
    \node[vertex] at (-90+13*\s:\r) {};

    \node[vertex, label={[shift={(-0.3,-0.25)}]\scriptsize $j$}] at (-90+14*\s:\r) {};
    \node[rotate=15*\s] at (-90+15*\s:\r) {\tiny$\cdots$};
    \node[vertex] at (-90+16*\s:\r) {};

\end{tikzpicture}%
        \caption{
            Depicted are three partitions with their induced cycle partitions.
            The three vectors $\psi$ on the right-hand side of \eqref{eq:construction-n3} are given by the difference of the feasible solutions induced by the partitions and the cycle partitions.
        }
        \label{fig:n3-unit-vectors}
    \end{figure}
    
    For any $j \in \set{2, \dots, q-1}$ the coefficient $b_\set{i,i+sq+j}$ can be constructed by the following combination illustrated in \Cref{fig:long-chord-unit-vectors}:
    \begin{equation}
    \begin{alignedat}{5} \label{eq:construction-long-chords}
        b_\set{i,i+sq+j}
        &= b^\top \psi(i-q+1,\ (\overline{q}, \ &&q, \dots, q, \ j-1, \ &&\overline{q-j+2},&&\ \ \,q, \dots, q&&)) \\
        &- b^\top \psi(i-q+1,\ (\overline{q-1},\ &&q, \dots, q, \ j,\ &&\overline{q-j+2},&&\ \ \, q, \dots, q&&)) \\ 
        &- b^\top \psi(i-q+1,\ (\overline{q},\ &&q, \dots, q, \ j,\ &&\overline{q-j+1},&&\ \ \, q, \dots, q&&)) \\
        &+ b^\top \psi(i-q+1,\ (\overline{q-1},\ &&\underbrace{q, \dots, q}_{\text{$s$ times}}, \ j+1,\ &&\overline{q-j+1},&&\underbrace{q, \dots, q}_{\text{$m{-}s{-}2$ times}}&&))\enspace .
    \end{alignedat}
    \end{equation}
    \\
    As the right-hand side evaluates to $0$, it follows $b_\set{i,i+sq+j} = 0$.
    Using the same construction while fixing $j=2$ and exchanging the position of component $s+2$ and $s+3$, $b_\set{i,i+sq+q}$ can be constructed: 
    \begin{alignat*}{5}
        b_\set{i,i+sq+q} 
        &= b^\top \psi(i-q+1,\ (\overline{q}, \ &&q, \dots, q,\ \overline{q},\ &&1, &&\ \ \, q, \dots, q&&)) \\
        &- b^\top \psi(i-q+1,\ (\overline{q-1},\ &&q, \dots, q,\ \overline{q},\ &&2, &&\ \ \, q, \dots, q&&)) \\ 
        &- b^\top \psi(i-q+1,\ (\overline{q},\ &&q, \dots, q,\ \overline{q-1},\ &&2, &&\ \ \, q, \dots, q&&)) \\
        &+ b^\top \psi(i-q+1,\ (\overline{q-1},\ &&\underbrace{q, \dots, q}_{\text{$s$ times}},\ \overline{q-1},\ &&3, &&\underbrace{q, \dots, q}_{\text{$m{-}s{-}2$ times}}&&)) \enspace .
    \end{alignat*}
    As the right-hand side evaluates again to $0$, it follows $b_\set{i,i+sq+q} = 0$.
    Lastly, it holds
    \begin{align*}
        \sum_{j=1}^{q} b_\set{i,i+sq+j}
        =  b^\top \psi(i,\ (\overline{1},\ \underbrace{q, \dots, q}_{\text{$s$ times}},\ \overline{q},\underbrace{q, \dots, q}_{\text{$m{-}s{-}2$ times}})) \enspace .
    \end{align*}
    Since all coefficients in the sum for $j = 2,\dots,q$ are $0$, and the right-hand side equals $0$ as well, it follows $b_\set{i,i+sq+1} = 0$.
    This determines the last class of coefficients and completes the proof.  
\end{proof}

\begin{remark}
    We observe that the $q$-chorded $k$-cycle inequalities with $k = 1$ mod $q$ appear in pairs.
    More specifically, let $p$ such that $k = pq + 1$.
    Then the $q$-chords $\{i,i+q\}$ for $i \in \Z_k$ form a cycle of length $k$, and the edges $\{i,i+1\}$ for $i \in \Z_k$ are precisely the $p$-chords of that cycle.
    Therefore, we obtain the following inequalities
    \[
        - k + q + 1 \leq \sum_{i \in \Z_k} \left(x_{\{i,i+1\}} - x_{\{i,i+q\}} \right) \leq k - p - 1\enspace ,
    \]
    where the first inequality is a $p$-chorded $k$-cycle inequality with respect to the cycle defined by the edges $\{i,i+q\}$ for $i \in \Z_k$ and the second inequality is just the regular $q$-chorded $k$-cycle inequality \eqref{eq:q-chored-k-cycle}.
    This generalizes the observation made in \cite{andres2023polyhedral} that for odd $k$, the $2$- and $\tfrac{k-1}{2}$-chorded cycle inequalities appear in pairs.
\end{remark}

\begin{figure}[t]
    \centering
    \begin{tikzpicture}
    \def \r{1.25cm}
    \def \rs{0.05*\r}
    \def \l{0.9cm}
    \def \s{16.3636}
    \def \d{1}

    \draw[arealight2] plot [smooth cycle, tension=0.2] coordinates {
		(-90-3*\s+\d:\r) 
		(-90+0*\s-\d:\r) 
		(-90+5*\s:\rs) 
		(-90+10*\s+2*\d:\r) 
		(-90+13*\s-\d:\r) 
		(-90+0*\s:\rs) 
    };

    \draw[area3] (-90+10*\s+\d:\r) arc (-90+10*\s+\d:-90+13*\s-\d:\r);
    \draw[area3] (-90-\d:\r) arc (-90-\d:-90-3*\s+\d:\r);

    \foreach \style in {area2}
    {
        \draw[\style] (-90-\d:\r) arc (-90-\d:-90-3*\s+\d:\r);
        \draw[\style] (-90+\s+\d:\r) arc (-90+\s+\d:-90+2*\s:\r);
        \draw[\style] (-90+4*\s:\r) arc (-90+4*\s:-90+6*\s:\r);
        \draw[\style] (-90+7*\s:\r) arc (-90+7*\s:-90+9*\s-\d:\r);
        \draw[\style] (-90+10*\s+\d:\r) arc (-90+10*\s+\d:-90+13*\s-\d:\r);
        \draw[\style] (-90+14*\s+\d:\r) arc (-90+14*\s+\d:-90+15*\s:\r);
        \draw[\style] (-90+17*\s:\r) arc (-90+17*\s:-90+18*\s-\d:\r);
	 }

    \node at (-90-1.5*\s:\l) {\scriptsize $q$};
    \node at (-90+1.5*\s:\l) {\scriptsize $q$};
    \node at (-90+5*\s:\l) {\scriptsize $q$};
    \node[rotate=180+8*\s] at (-90+8*\s:\l) {\scriptsize $j\hspace{-2pt}-\hspace{-2pt}1$};
    \node[rotate=180+11.3*\s] at (-90+11.5*\s:\l) {\scriptsize $q\hspace{-2pt}-\hspace{-2pt}j\hspace{-2pt}+\hspace{-2pt}2$};
    \node at (-90+14.5*\s:\l) {\scriptsize $q$};
    \node at (-90+17.5*\s:\l) {\scriptsize $q$};

    \node[vertex, label={[shift={(0.0,-0.6)}]\scriptsize $i$}] at (-90:\r) {};
    \node[vertex] at (-90+\s:\r) {};
    \node[rotate=2*\s] at (-90+2*\s:\r) {\tiny $\cdots$};
    \node[rotate=3*\s] at (-90+3*\s:\r) {\tiny$\cdots$};
    \node[rotate=4*\s] at (-90+4*\s:\r) {\tiny$\cdots$};
    \node[vertex] at (-90+5*\s:\r) {};
    \node[vertex] at (-90+6*\s:\r) {};
    \node[vertex] at (-90+7*\s:\r) {};
    \node[rotate=8*\s] at (-90+8*\s:\r) {\tiny$\cdots$};
    \node[vertex] at (-90+9*\s:\r) {};
    \node[vertex,label={[shift={(0.0,0.1)}]\scriptsize $i\hspace{-2pt}+\hspace{-2pt}sq\hspace{-2pt}+\hspace{-2pt}j$}] at (-90+10*\s:\r) {};
    \node[vertex] at (-90+11*\s:\r) {};
    \node[rotate=12*\s] at (-90+12*\s:\r) {\tiny$\cdots$};
    \node[vertex] at (-90+13*\s:\r) {};
    \node[vertex] at (-90+14*\s:\r) {};
    \node[rotate=15*\s] at (-90+15*\s:\r) {\tiny$\cdots$};
    \node[rotate=16*\s] at (-90+16*\s:\r) {\tiny$\cdots$};
    \node[rotate=17*\s] at (-90+17*\s:\r) {\tiny$\cdots$};
    \node[vertex] at (-90+18*\s:\r) {};

    \node[vertex] at (-90-\s:\r) {};
    \node[rotate=-2*\s] at (-90-2*\s:\r) {\tiny$\cdots$};
    \node[vertex] at (-90-3*\s:\r) {};
\end{tikzpicture}%
    \hfill
    \begin{tikzpicture}
    \def \r{1.25cm}
    \def \rs{0.05*\r}
    \def \l{0.9cm}
    \def \s{16.3636}
    \def \d{1}

    \draw[arealight2] plot [smooth cycle, tension=0.2] coordinates {
		(-90-3*\s+\d:\r) 
		(-90-1*\s-\d:\r) 
		(-90+5*\s:\rs) 
		(-90+10*\s+2*\d:\r) 
		(-90+13*\s-\d:\r)
		(-90+0*\s:\rs) 
    };

    \draw[area3] (-90-1*\s-\d:\r) arc (-90-1*\s-\d:-90-3*\s+\d:\r);
    \draw[area3] (-90+10*\s+\d:\r) arc (-90+10*\s+\d:-90+13*\s-\d:\r);

    \draw[area2] (-90-1*\s-\d:\r) arc (-90-1*\s-\d:-90-3*\s+\d:\r);
    \node[rotate=-2*\s] at (-90-2*\s:\l) {\scriptsize $q\hspace{-2pt}-\hspace{-2pt}1$};
    \draw[area2] (-90+0*\s+\d:\r) arc (-90+0*\s+\d:-90+2*\s:\r);
    \node at (-90+1*\s:\l) {\scriptsize $q$};
    \draw[area2] (-90+4*\s:\r) arc (-90+4*\s:-90+5*\s:\r);
    \node at (-90+4.5*\s:\l) {\scriptsize $q$};
    \draw[area2] (-90+6*\s:\r) arc (-90+6*\s:-90+9*\s-\d:\r);
    \node at (-90+7.5*\s:\l) {\scriptsize $j$};
    \draw[area2] (-90+10*\s+\d:\r) arc (-90+10*\s+\d:-90+13*\s-\d:\r);
    \node[rotate=180+11.3*\s] at (-90+11.5*\s:\l) {\scriptsize $q\hspace{-2pt}-\hspace{-2pt}j\hspace{-2pt}+\hspace{-2pt}2$};
    \draw[area2] (-90+14*\s+\d:\r) arc (-90+14*\s+\d:-90+15*\s:\r);
    \node at (-90+14.5*\s:\l) {\scriptsize $q$};
    \draw[area2] (-90+17*\s:\r) arc (-90+17*\s:-90+18*\s:\r);
    \node at (-90+17.5*\s:\l) {\scriptsize $q$};

    \node[vertex, label={[shift={(0.0,-0.6)}]\scriptsize $i$}] at (-90:\r) {};
    \node[vertex] at (-90+\s:\r) {};
    \node[rotate=2*\s] at (-90+2*\s:\r) {\tiny $\cdots$};
    \node[rotate=3*\s] at (-90+3*\s:\r) {\tiny$\cdots$};
    \node[rotate=4*\s] at (-90+4*\s:\r) {\tiny$\cdots$};
    \node[vertex] at (-90+5*\s:\r) {};
    \node[vertex] at (-90+6*\s:\r) {};
    \node[vertex] at (-90+7*\s:\r) {};
    \node[rotate=8*\s] at (-90+8*\s:\r) {\tiny$\cdots$};
    \node[vertex] at (-90+9*\s:\r) {};
    \node[vertex,label={[shift={(0.0,0.1)}]\scriptsize $i\hspace{-2pt}+\hspace{-2pt}sq\hspace{-2pt}+\hspace{-2pt}j$}] at (-90+10*\s:\r) {};
    \node[vertex] at (-90+11*\s:\r) {};
    \node[rotate=12*\s] at (-90+12*\s:\r) {\tiny$\cdots$};
    \node[vertex] at (-90+13*\s:\r) {};
    \node[vertex] at (-90+14*\s:\r) {};
    \node[rotate=15*\s] at (-90+15*\s:\r) {\tiny$\cdots$};
    \node[rotate=16*\s] at (-90+16*\s:\r) {\tiny$\cdots$};
    \node[rotate=17*\s] at (-90+17*\s:\r) {\tiny$\cdots$};
    \node[vertex] at (-90+18*\s:\r) {};

    \node[vertex] at (-90-\s:\r) {};
    \node[rotate=-2*\s] at (-90-2*\s:\r) {\tiny$\cdots$};
    \node[vertex] at (-90-3*\s:\r) {};
\end{tikzpicture}%
    \hfill
    \begin{tikzpicture}
    \def \r{1.25cm}
    \def \rs{0.05*\r}
    \def \l{0.9cm}
    \def \s{16.3636}
    \def \d{1}

    \draw[arealight2] plot [smooth cycle, tension=0.2] coordinates {
		(-90-3*\s+\d:\r) 
		(-90+0*\s-\d:\r) 
		(-90+5*\s:\rs) 
		(-90+11*\s+2*\d:\r) 
		(-90+13*\s-\d:\r) 
		(-90+0*\s:\rs) 
	};
    
    \draw[area3] (-90-\d:\r) arc (-90-\d:-90-3*\s+\d:\r);
    \draw[area3] (-90+11*\s+\d:\r) arc (-90+11*\s+\d:-90+13*\s-\d:\r);

    \draw[area2] (-90-\d:\r) arc (-90-\d:-90-3*\s+\d:\r);
    \node at (-90-1.5*\s:\l) {\scriptsize $q$};
    \draw[area2] (-90+\s+\d:\r) arc (-90+\s:-90+2*\s:\r);
    \node at (-90+1.5*\s:\l) {\scriptsize $q$};
    \draw[area2] (-90+4*\s:\r) arc (-90+4*\s:-90+6*\s:\r);
    \node at (-90+5*\s:\l) {\scriptsize $q$};
    \draw[area2] (-90+7*\s:\r) arc (-90+7*\s:-90+10*\s-\d:\r);
    \node at (-90+8.5*\s:\l) {\scriptsize $j$};
    \draw[area2] (-90+11*\s+\d:\r) arc (-90+11*\s+\d:-90+13*\s-\d:\r);
    \node[rotate=180+11.8*\s] at (-90+12*\s:\l) {\scriptsize $q\hspace{-2pt}-\hspace{-2pt}j\hspace{-2pt}+\hspace{-2pt}1$};
    \draw[area2] (-90+14*\s+\d:\r) arc (-90+14*\s+\d:-90+15*\s:\r);
    \node at (-90+14.5*\s:\l) {\scriptsize $q$};
    \draw[area2] (-90+17*\s:\r) arc (-90+17*\s:-90+18*\s-\d:\r);
    \node at (-90+17.5*\s:\l) {\scriptsize $q$};

    \node[vertex, label={[shift={(0.0,-0.6)}]\scriptsize $i$}] at (-90:\r) {};
    \node[vertex] at (-90+\s:\r) {};
    \node[rotate=2*\s] at (-90+2*\s:\r) {\tiny $\cdots$};
    \node[rotate=3*\s] at (-90+3*\s:\r) {\tiny$\cdots$};
    \node[rotate=4*\s] at (-90+4*\s:\r) {\tiny$\cdots$};
    \node[vertex] at (-90+5*\s:\r) {};
    \node[vertex] at (-90+6*\s:\r) {};
    \node[vertex] at (-90+7*\s:\r) {};
    \node[rotate=8*\s] at (-90+8*\s:\r) {\tiny$\cdots$};
    \node[vertex] at (-90+9*\s:\r) {};
    \node[vertex,label={[shift={(0.0,0.1)}]\scriptsize $i\hspace{-2pt}+\hspace{-2pt}sq\hspace{-2pt}+\hspace{-2pt}j$}] at (-90+10*\s:\r) {};
    \node[vertex] at (-90+11*\s:\r) {};
    \node[rotate=12*\s] at (-90+12*\s:\r) {\tiny$\cdots$};
    \node[vertex] at (-90+13*\s:\r) {};
    \node[vertex] at (-90+14*\s:\r) {};
    \node[rotate=15*\s] at (-90+15*\s:\r) {\tiny$\cdots$};
    \node[rotate=16*\s] at (-90+16*\s:\r) {\tiny$\cdots$};
    \node[rotate=17*\s] at (-90+17*\s:\r) {\tiny$\cdots$};
    \node[vertex] at (-90+18*\s:\r) {};

    \node[vertex] at (-90-\s:\r) {};
    \node[rotate=-2*\s] at (-90-2*\s:\r) {\tiny$\cdots$};
    \node[vertex] at (-90-3*\s:\r) {};
\end{tikzpicture}%
    \hfill
    \begin{tikzpicture}
    \def \r{1.25cm}
    \def \rs{0.05*\r}
    \def \l{0.9cm}
    \def \s{16.3636}
    \def \d{1}

    \draw[arealight2] plot [smooth cycle, tension=0.2] coordinates {
	    (-90-3*\s+\d:\r) 
	    (-90-1*\s-\d:\r) 
	    (-90+5*\s:\rs) 
	    (-90+11*\s+2*\d:\r) 
	    (-90+13*\s-\d:\r) 
	    (-90+0*\s:\rs) 
    };
    
    \draw[area2] (-90-1*\s-\d:\r) arc (-90-1*\s-\d:-90-3*\s+\d:\r);
    \node[rotate=-2*\s] at (-90-2*\s:\l) {\scriptsize $q\hspace{-2pt}-\hspace{-2pt}1$};
    \draw[area2] (-90+0*\s+\d:\r) arc (-90+0*\s+\d:-90+2*\s:\r);
    \node at (-90+1*\s:\l) {\scriptsize $q$};
    \draw[area2] (-90+4*\s:\r) arc (-90+4*\s:-90+5*\s:\r);
    \node at (-90+4.5*\s:\l) {\scriptsize $q$};
    \draw[area2] (-90+6*\s:\r) arc (-90+6*\s:-90+10*\s-\d:\r);
    \node[rotate=180+8*\s] at (-90+8*\s:\l) {\scriptsize $j\hspace{-2pt}+\hspace{-2pt}1$};
    \draw[area2] (-90+11*\s+\d:\r) arc (-90+11*\s+\d:-90+13*\s-\d:\r);
    \node[rotate=180+11.8*\s] at (-90+12*\s:\l) {\scriptsize $q\hspace{-2pt}-\hspace{-2pt}j\hspace{-2pt}+\hspace{-2pt}1$};
    \draw[area2] (-90+14*\s+\d:\r) arc (-90+14*\s+\d:-90+15*\s:\r);
    \node at (-90+14.5*\s:\l) {\scriptsize $q$};
    \draw[area2] (-90+17*\s:\r) arc (-90+17*\s:-90+18*\s-\d:\r);
    \node at (-90+17.5*\s:\l) {\scriptsize $q$};

    \node[vertex, label={[shift={(0.0,-0.6)}]\scriptsize $i$}] at (-90:\r) {};
    \node[vertex] at (-90+\s:\r) {};
    \node[rotate=2*\s] at (-90+2*\s:\r) {\tiny $\cdots$};
    \node[rotate=3*\s] at (-90+3*\s:\r) {\tiny$\cdots$};
    \node[rotate=4*\s] at (-90+4*\s:\r) {\tiny$\cdots$};
    \node[vertex] at (-90+5*\s:\r) {};
    \node[vertex] at (-90+6*\s:\r) {};
    \node[vertex] at (-90+7*\s:\r) {};
    \node[rotate=8*\s] at (-90+8*\s:\r) {\tiny$\cdots$};
    \node[vertex] at (-90+9*\s:\r) {};
    \node[vertex,label={[shift={(0.0,0.1)}]\scriptsize $i\hspace{-2pt}+\hspace{-2pt}sq\hspace{-2pt}+\hspace{-2pt}j$}] at (-90+10*\s:\r) {};
    \node[vertex] at (-90+11*\s:\r) {};
    \node[rotate=12*\s] at (-90+12*\s:\r) {\tiny$\cdots$};
    \node[vertex] at (-90+13*\s:\r) {};
    \node[vertex] at (-90+14*\s:\r) {};
    \node[rotate=15*\s] at (-90+15*\s:\r) {\tiny$\cdots$};
    \node[rotate=16*\s] at (-90+16*\s:\r) {\tiny$\cdots$};
    \node[rotate=17*\s] at (-90+17*\s:\r) {\tiny$\cdots$};
    \node[vertex] at (-90+18*\s:\r) {};

    \node[vertex] at (-90-\s:\r) {};
    \node[rotate=-2*\s] at (-90-2*\s:\r) {\tiny$\cdots$};
    \node[vertex] at (-90-3*\s:\r) {};
\end{tikzpicture}%
    \caption{
        Depicted are four partitions with their induced cycle partitions.
        For clarity, we show in contrast to previous figures only the cycle partition and the component of the inducing partition that is distinct from the cycle partition.
        The four vectors $\psi$ on the right-hand side of \eqref{eq:construction-long-chords} are given by the difference of the feasible solutions induced by the partitions and the cycle partitions.
        }
    \label{fig:long-chord-unit-vectors}
\end{figure}

\section{Conclusion}
We establish exact conditions under which the $q$-chorded $k$-cycle inequalities described by Müller and Schulz in 2002 induce facets of the clique partitioning polytope.
For $q \in \{2, \tfrac{k-1}{2}\}$, these conditions specialize to properties previously known.
In their general form, they imply the existence of many facets induced by $q$-chorded $k$-cycle inequalities for $2 < q < \tfrac{k-1}{2}$ previously unknown.
The conditions under which chorded cycle inequalities do \emph{not} induce facets are particularly interesting because the faces induced by such inequalities are contained in other facets currently unknown.

\section*{Acknowledgements}
This work is partly supported by BMFTR (Federal Ministry of Research, Technology and Space) in DAAD project 57616814 (SECAI, School of Embedded Composite AI, \url{https://secai.org/}) as part of the program Konrad Zuse Schools of Excellence in Artificial Intelligence.

\bibliographystyle{abbrvurl}
\bibliography{references}

\appendix
\section{Appendix}

\subsection{Proof of \Cref{lem:valid}}\label{app:proof-valid}

We prove that the $q$-chorded $k$-cycle inequality \eqref{eq:q-chored-k-cycle} is valid for the clique partitioning polytope by showing that it can be obtained by a non-negative linear combination of triangle inequalities \eqref{eq:triangle} and box inequalities \eqref{eq:box}, and rounding down the right-hand side.
In particular, this shows that the $q$-chorded $k$-cycle inequalities are Chv\'atal-Gomory cuts of the system of triangle and box inequalities.

For $i \in \Z_k$ the inequality
\begin{align}\label{eq:cycle-inequality}
    \sum_{j = 0}^{q-1} x_{\{i+j, i+j+1\}} - x_{\{i, i+q\}} \leq q - 1
\end{align}
is a so-called \emph{cycle inequality} and is known to be valid for the clique partitioning problem \cite{chopra1993partition}.
In fact, \eqref{eq:cycle-inequality} is obtained by summing the triangle inequalities
\[
    x_{\{i, i+j\}} + x_{\{i+j, i+j+1\}} - x_{\{i, i+j+1\}} \leq 1
\]
for $j = 1,\dots,q-1$.
By adding $(q-1)$ times the constraint $-x_{\{i, i+q\}} \leq 0$ to \eqref{eq:cycle-inequality} we obtain the valid inequality
\begin{align}\label{eq:relaxed-cycle-inequality}
    \sum_{j = 0}^{q-1} x_{\{i+j, i+j+1\}} - q x_{\{i, i+q\}} \leq q - 1
\end{align}
which we call a \emph{relaxed cycle inequality}.
By summing relaxed cycle inequalities \eqref{eq:relaxed-cycle-inequality} for all $i \in \Z_k$ we obtain
\begin{align}\label{eq:summed-cycle-inequalities}
    \sum_{i \in Z_k} \left( qx_{\{i,i+1\}} - qx_{\{i, i+q\}} \right) \leq k (q-1) \enspace .
\end{align}
Dividing \eqref{eq:summed-cycle-inequalities} by $q$ and rounding the right-hand side down yields the desired result with
\begin{align*}
    \left\lfloor \frac{k (q-1)}{q} \right\rfloor
    = \left\lfloor k - \frac{k}{q} \right\rfloor
    = k - \left\lceil \frac{k}{q} \right\rceil \enspace .
\end{align*}
\qed

\subsection{Proof of \Cref{lem:part}}\label{app:proof-part}
We prove that for any $q, k \in \N$ with $2 \leq q \leq k / 2$ the feasible solution $x^\Pi$ associated with a partition $\Pi$ of $\Z_k$ satisfies \eqref{eq:q-chored-k-cycle} at equality if and only if \ref{item:eq-a} or \ref{item:eq-b} is fulfilled.

First, we note that the number of edges $\set{\set{i,i+1} \mid i \in \Z_k}$ whose nodes are in the same set of the inducing partition is equal to the number of nodes minus the number of sets in the cycle partition:
\begin{align}\label{eq:cycle_edges}
    \sum_{i \in \Z_k} x_{\{i,i+1\}} = k - |\Pi^{\cyc}|\enspace .
\end{align}
This relation will be essential for showing sufficiency and necessity of the specified conditions.

We first show sufficiency.
Assume \ref{item:eq-a} is fulfilled.
As $|\Pi^{\cyc}|  = \lceil k / q \rceil$ by \ref{item:eq-a1}, we have $\sum_{i \in \Z_k} x_{\{i,i+1\}} = k - \lceil k / q \rceil$ by \eqref{eq:cycle_edges}.
Furthermore, by \ref{item:eq-a2}, all sets in the cycle partition are of size less than or equal to $q$, and by \ref{item:eq-a3}, there lie at least $q$ nodes between distinct sets of the cycle partition that are subsets of the same set in the inducing partition. Thus, $\sum_{i \in \Z_k} x_{\{i, i+q\}} = 0$.
Consequently, \eqref{eq:q-chored-k-cycle} is fulfilled at equality.

Assume now that \ref{item:eq-b} is fulfilled.
As $|\Pi^{\cyc}|  = \lceil k / q \rceil - 1$ by \ref{item:eq-b1}, we have $\sum_{i \in \Z_k} x_{\{i,i+1\}} = k - \lceil k / q \rceil + 1$ by \eqref{eq:cycle_edges}.
Furthermore, by \ref{item:eq-b2}, one set of the cycle partition is of size $q+1$, all other sets are of size $q$, and consequently, there also lie at least $q$ nodes between distinct sets of the cycle partition that are subsets of the same set in the inducing partition. Thus, $\sum_{i \in \Z_k} x_{\{i, i+q\}} = 1$.
Consequently, \eqref{eq:q-chored-k-cycle} is fulfilled at equality.
This finishes the proof of sufficiency.

We now show necessity. 
Let $x^\Pi$ be the feasible solution associated with a partition $\Pi$ of $\Z_k$ that satisfies \eqref{eq:q-chored-k-cycle} at equality.
By \eqref{eq:cycle_edges}, we must only show that \ref{item:eq-a} or \ref{item:eq-b} holds if 
\begin{align}\label{eq:equality_simplified}
    \sum_{i \in \Z_k} x_{\{i, i+q\}} = \left\lceil \frac{k}{q} \right\rceil - |\Pi^{\cyc}|\enspace .
\end{align}
We first show that \ref{item:eq-a} is satisfied if $\sum_{i \in \Z_k} x_{\{i, i+q\}} = 0$ and then that \ref{item:eq-b} is satisfied if $\sum_{i \in \Z_k} x_{\{i, i+q\}} > 0$.

If $\sum_{i \in \Z_k} x_{\{i, i+q\}} = 0$, we directly get $|\Pi^{\cyc}| = \lceil \frac{k}{q} \rceil$ by \eqref{eq:equality_simplified}, so \ref{item:eq-a1} is satisfied.
Further, $\sum_{i \in \Z_k} x_{\{i, i+q\}} = 0$ implies $x_{\{i, i+q\}} = 0$ for all $i \in \Z_k$, so there is no $q$-chord whose nodes are in the same set.
Consequently, all sets must be smaller than or equal to $q$, i.e.~$|U| \leq q$ for all $U \in \Pi^{\cyc}$, and \ref{item:eq-a2} holds.
Another implication of this is that distinct sets of $\Pi^{\cyc}$ that are subsets of the same set of $\Pi$ must be separated by at least $q$ nodes, i.e.~\ref{item:eq-a3} holds.
Thus, \ref{item:eq-a} is satisfied for $\sum_{i \in \Z_k} x_{\{i, i+q\}} = 0$.

Assume now that $\sum_{i \in \Z_k} x_{\{i, i+q\}} > 0$ and thus, by \eqref{eq:equality_simplified}, that $\lceil \frac{k}{q} \rceil > |\Pi^{\cyc}|$.
If $x_{\{i, i+q\}} = 1$ for some $i \in \Z_k$, there either exists a set in the cycle partition containing nodes $i$ and $i+q$ that is of size at least $q+1$, or there exist two distinct sets in the cycle partition that contain these nodes and are subsets of the same set of the inducing partition.
With $m = \sum_{i \in \mathbb{N}} i \left|\left\{U \in \Pi^{\cyc} \mid |U| = q+i\right\}\right|$ denoting the number of nodes that exceed the size of $q$ in their set, we thus get 
\begin{align}\label{eq:geq-n}
    \sum_{i \in \Z_k} x_{\{i, i+q\}} \geq m \enspace .
\end{align}
Equality holds if all $q$-chords connected in the inducing partition are also connected in the cycle partition.
Furthermore, the number of sets in the cycle partition is bounded,
\begin{align}\label{eq:k-n}
    |\Pi^{\cyc}| \geq \left\lceil\frac{k - m}{q}\right\rceil\enspace ,
\end{align}
since the remaining $k-m$ nodes in the cycle can form sets of at most size $q$.
Combining \eqref{eq:equality_simplified} with \eqref{eq:geq-n} and \eqref{eq:k-n}, we get
\begin{align}
    \left\lceil\frac{k - m}{q}\right\rceil + m \leq  \left\lceil \frac{k}{q} \right\rceil\enspace .
\end{align}
Clearly, this inequality can only be satisfied for either $m = 0$, or $m=1$ and $k = 1$ mod $q$. 
Moreover, as the inequality $\lceil \frac{k}{q} \rceil > |\Pi^{\cyc}| \geq \lceil\frac{k - m}{q}\rceil$ is strict, only the case with $m=1$ is possible. 
Thus, $|\Pi^{\cyc}| = \lceil\frac{k - m}{q}\rceil = \lceil\frac{k}{q}\rceil - 1$, i.e.~\ref{item:eq-b1} holds.
Furthermore, as $m = 1$, there exists exactly one set $U \in \Pi^{\cyc}$ of size $q+1$ and no set in the cycle partition of size greater than $q+1$. 
Noting that all nodes must be in some set of the cycle partition, $k = \sum_{U' \in \Pi^{\cyc}} |U'|$ and the already established fact that $|\Pi^{\cyc}| = \lceil\frac{k}{q}\rceil - 1$, we get $|U'| = q$ for all $U' \in \Pi^{\cyc} \setminus \{U\}$, i.e.~\ref{item:eq-b2} holds.
This completes the proof.
\qed

\end{document}